\documentclass[sigconf,nonacm]{acmart}

\usepackage{tikz}
\usepackage{amsmath,amsfonts}
\usepackage[linesnumbered,vlined]{algorithm2e}
\usepackage{graphicx}
\usepackage{textcomp}
\usepackage{xcolor}
\usepackage{xspace}
\usepackage{tabularx}
\usepackage{makecell}
\usepackage{paralist}
\usepackage{caption}
\usepackage{subcaption}
\usepackage{soul}
\usepackage{mdframed}
\usepackage[inline,shortlabels]{enumitem}
\usepackage{balance}

\newcolumntype{L}[1]{>{\raggedright\arraybackslash}p{#1}}
\newcolumntype{C}[1]{>{\centering\arraybackslash}p{#1}}
\newcolumntype{R}[1]{>{\raggedleft\arraybackslash}p{#1}}

\makeatletter
\newtheorem*{rep@theorem}{\rep@title}
\newcommand{\newreptheorem}[2]{%
\newenvironment{rep#1}[1]{%
 \def\rep@title{#2 \ref{##1}}%
 \begin{rep@theorem}}%
 {\end{rep@theorem}}}
\makeatother

\newtheorem{theorem}{Theorem}
\newtheorem{lemma}{Lemma}
\newreptheorem{lemma}{Lemma}

\newcommand{\thesystem}{DuoBFT\xspace}
\newcommand{\themcsystem}{MC-DuoBFT\xspace}

\begin{document}

\title{\thesystem: Resilience vs. Performance Trade-off in Byzantine Fault Tolerance}

\author{Balaji Arun}
\affiliation{%
  \institution{Virginia Tech}
}
\email{balajia@vt.edu}

\author{Binoy Ravindran}
\affiliation{%
  \institution{Virginia Tech}
}
\email{binoy@vt.edu}

\begin{abstract}
This paper presents \thesystem, a Byzantine fault-tolerant protocol that 
uses trusted components to provide commit decisions in the Hybrid fault model 
in addition to commit decisions in the BFT model.
By doing so, it enables the clients to choose the response fault model for 
its commands.
Internally, \thesystem commits each client command under both the hybrid 
and Byzantine models, but since hybrid commits take fewer communication steps 
and use smaller quorums than BFT commits, clients can benefit from the 
low-latency commits in the hybrid model.

\thesystem uses a common view-change change protocol to handle both fault models.
To achieve this, we enable a notion called Flexible Quorums in the hybrid fault 
model by revisiting the quorum intersection requirements in hybrid protocols. 
The flexible quorum technique enables having a hybrid view change quorum 
that is of the same size as a BFT view-change quorum.
This paves a path for efficiently combining both the fault models 
within a single unified protocol.
Our evaluation on a wide-area deployment reveal that \thesystem can provide 
hybrid commits with 30\% lower latency to existing protocols without 
sacrificing throughput. 
In absolute terms, \thesystem provides sub-200-millisecond latency in a 
geographically replicated deployment.
\end{abstract}

\maketitle

\section{Introduction}

Byzantine fault-tolerant (BFT) protocols are a building block of 
many decentralized ledger or Blockchain 
systems~\cite{Baudet:2019:LibraBFT:TechRep,Buchman:2016:Tendermint:Thesis}.
Traditionally, the design of a BFT protocol is driven by a
set of assumptions that includes the timing model (synchrony, partial 
synchrony, or asynchrony) and the failure model (BFT, 
Hybrids~\cite{Veronese:2013:Efficient:TC.2011.221}, 
XFT~\cite{Liu:2016:XFT:OSDI:3026877.3026915}).
These assumptions not only serve as the basis for safety and liveness guarantees,
but also establishes performance expectations.

The BFT failure model, being the most general one, allows any arbitrary 
behavior in the system that deviates from the protocol specification. 
The PBFT protocol~\cite{Castro:1999:PBFT:OSDI:296806.296824}, 
a partially synchronous protocol in this fault model,
can optimally tolerate less than $1/3$ failures and requires three communication 
delays (not including client communication) to reach commitment among at least
$2/3$ replicas. 
Alternatively, the Hybrid fault model infuses a BFT protocol 
with trusted assumptions allowing construction of protocols that tolerate 
less than $1/2$ failures and require only two communication delays to reach 
commitment among at least $1/2$ replicas. 

The ability to tolerate 50\% more failures than BFT 
protocols with fewer communication delays and with only majority quorums  
make Hybrid protocols appealing for building low-latency Blockchain systems
(e.g. Point of Sale Payment Systems~\cite{Baudet:2020:FastPay:AFT:3419614.3423249}).
Hybrid protocols require trusted components within each replica to 
prevent equivocating behavior of malicious replicas. 
Modern commodity processors have special mechanisms 
(e.g. Intel SGX~\cite{Costan:2016:SGX:TR}, 
AMD SEV~\cite{Mofrad:2018:SGX-AMD:HASP:3214292.3214301}, 
and ARM TrustZone~\cite{Sandro:2019:ArmTrustZone:3291047}) to implement these trusted 
components in software via Trusted Execution Environments (TEE) 
that is isolated from other parts of the system without any additional hardware. 

The use of trusted execution environment raises some challenges. First, it 
greatly reduces the choice of hardware 
(e.g. Intel SGX, ARM TrustZone, AMD SEV) 
used to deploy such protocols. 
However, Byzantine protocols implicitly/explicitly require diversity in 
the deployment stack to reduce the number of correlated failures
\cite{Garcia:2019:Lazarus:3361525.3361550}. 
Furthermore, security vulnerabilities have been discovered in trusted execution 
environments recently~\cite{Schaik:2020:SGAxe,Schaik:2020:Cacheout:13353,Li:2021:CIPHERLEAKS:UsenixSecurity}. 
Although active research in the area aims to solve these problems, 
the impact of undiscovered vulnerabilities raises concerns.
This raises questions on their applicability to BFT and 
Blockchain systems.

This paper presents \thesystem that encompasses a hybrid protocol and a BFT 
protocol in a single package, and enables the client to choose their response 
fault model. 
\thesystem always commits commands under both the fault models, 
ensuring that they have the best BFT safety guarantees, but since the Hybrid 
protocol is cheaper and quicker, clients can opt-in for Hybrid commit response.
This allows clients that require fast response to opt-in for Hybrid commit 
response, while clients that require higher resilience to wait for the 
BFT commit response. 
In doing so, \thesystem provides a unique trade-off to the clients: quick 
decisions made possible by hybrid replicas versus uncompromising resilience to 
malicious behavior. 
Furthermore, our solution allows different clients to individually adapt their 
fault assumptions dynamically without depending on the replicas.

A major contribution that make \thesystem possible is Flexible Hybrid 
Quorums. 
We show that the strictly majority quorums in hybrid protocols can be replaced 
with flexible intersecting quorums. 
Specifically, the commit agreement quorums need not intersect with each other, 
but only need to intersect with view change quorums. 
The net outcome is that our Hybrid quorums only require $f+1$ replicas for 
commit agreement and $N-f$ replicas for view-changes. 
We apply this technique to MinBFT~\cite{Veronese:2013:Efficient:TC.2011.221} and 
call the resulting protocol \emph{Flexible} MinBFT.

The quorum flexibility enables \thesystem to have a common view change protocols 
for both the Hybrid and BFT assumptions.
Since $N=3f+1$ for a BFT protocol, \thesystem provides Hybrid commits with 
$f+1$ quorum and BFT commits with a $2f+1$ quorum. The view change quorum is 
$2f+1$ for both the protocols. 

In \thesystem, the replicas propose and vote on blocks that contain client 
transactions or operations, and use separate commit rules for each fault model 
to make commit decisions on the blocks.
Replicas internally collect two types of quorums that form the basis of the 
commit rules:  
the Hybrid quorum consists of votes from $f+1$ replicas,
and the BFT quorum consists of votes from $2f+1$ replicas.
Collecting separate quorums allows \thesystem to tolerate any vulnerabilities 
affecting the Hybrid model. Specifically, the compromise of the trusted 
component only affects the Hybrid quorum but not the BFT quorum.

Furthermore, the flexibility provided by the \thesystem's dual fault model 
is better than 
speculation\cite{Gutea:2019:SBFT:DSN.2019.00063,Kotla:2007:Zyzzyva:SOSP:1294261.1294267}, or
tentative\cite{Castro:2001:PBFT:Thesis} execution capabilities
provided by other known protocols. 
While speculation requires 50\% larger quorums than 
PBFT~\cite{Castro:2002:PBFT:TOCS:571637.571640} and 
the tentative execution only reduces the execution overhead by overlapping 
the last communication step with execution, 
the hybrid model uses 50\% smaller quorums than PBFT and reduces one 
overall communication step.
At the same time, \thesystem does not incur more communication delays 
or larger quorums 
than that required to commit under the BFT model, unlike
\cite{Kotla:2007:Zyzzyva:SOSP:1294261.1294267,Gutea:2019:SBFT:DSN.2019.00063}.

We evaluate multiple variants of the \thesystem and show that it provides 
similar throughput to existing protocols while providing 
significantly lower latency. 
\themcsystem, which is optimization over \thesystem to use multiple 
instances, provides 30\% lower latency for 
clients expecting only hybrid commit responses.

This paper makes the following contributions:

\begin{compactitem}[-]
    \item \textbf{Flexible Hybrid Quorums}: We show that the use of majority 
quorums in the Hybrid fault model can be relaxed and replaced with simple 
intersecting quorums. This allows flexibility in the sizes of quorums used for 
different parts of the protocol.
    \item \textbf{Flexible MinBFT}: We apply Flexible Hybrid Quorums to 
MinBFT and present a protocol that uses $f+1$ quorums for commit agreement 
and larger $N-f$ quorums for view changes.
    \item \textbf{DuoBFT}: We present a BFT protocol under the partial synchrony 
timing model that can make Hybrid commit decisions with only $f+1$ replicas in 
addition to making traditional BFT commit decisions with $2f+1$ replicas. 
The protocol allows clients to choose their response fault model, making it 
possible for applications that require low latency to benefit from the Hybrid 
commits, while providing the ability to leverage traditional BFT guarantees.
\end{compactitem}

The rest of the paper is organized as follows. 
Section~\ref{sec:preliminaries} presents the terminology and system model. 
Section~\ref{sec:minbft} presents Flexible Hybrid Quorums and Flexible 
MinBFT. 
Section~\ref{sec:duobft} presents the \thesystem protocol, explanation of 
its properties along with proofs, and some optimizations. 
A discussion on some unique features of \thesystem with respect to 
existing solutions is presented in Section~\ref{sec:discussion}. 
Section~\ref{sec:eval} evaluates the protocols.
Section~\ref{sec:relwork} presents the related work and 
Section~\ref{sec:conclusion} concludes the paper.

\section{Preliminaries}
\label{sec:preliminaries}

In this section, we will discuss the necessary background for understanding the 
rest of the paper.

\subsection{Byzantine Consensus}
\label{sec:bg:byz}

\begin{figure}[t]
	\centering
	\begin{subfigure}{0.45\textwidth}
        \centering
		\includegraphics[width=\linewidth]{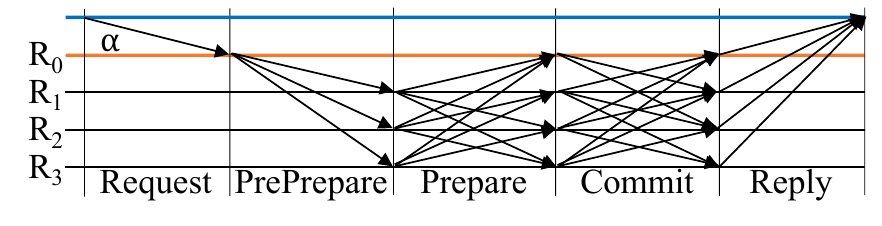}
		\caption{PBFT}
	\end{subfigure}
	\begin{subfigure}{0.35\textwidth}
        \centering
		\includegraphics[width=\linewidth]{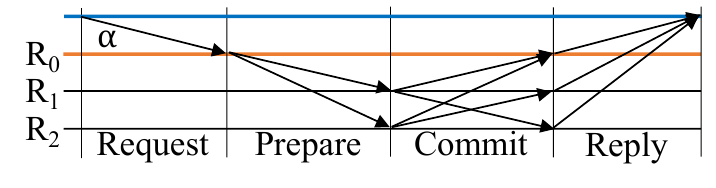}
		\caption{MinBFT}
	\end{subfigure}
	\caption{An example of normal execution steps in BFT and Hybrid fault models.}
	\label{fig:bg:walkthrough}
\end{figure}

A Byzantine Consensus protocol reaches agreement on the order of client-issued 
commands among a set of replicas some of which can be malicious. 
The commands are then executed in the agreed order on the shared 
state fully-replicated among the replicas.
A BFT protocol under the partial synchrony fault model can tolerate up to
$f$ malicious failures in a system of $N = 3f+1$ replicas. 
Most protocols are primary-based and proceed in a sequence of views, 
where in each view, a primary replica sequences client commands that every 
replica executes.
Correct replicas execute only after ensuring that a significant number of 
other replicas are aware of the same command and its execution order. 
To accomplish this, the protocol executes the Agreement subprotocol that 
involves exchanging the command and metadata among replicas. 
The number of communication phases differ by protocols.
PBFT~\cite{Castro:1999:PBFT:OSDI:296806.296824}, 
for example, uses three phases of communication to gather consent from 
a supermajority (i.e. 66\%) of replicas to decide the ordering for each command.

The View Change subprotocol is used to rotate the primary when the protocol 
is unable to make progress, i.e. order commands, with the current primary. 
A new view is installed after replicas exchange state information about the 
previous view and a new primary takes over. 
If the new primary does not make progress, another one takes its place.

Since the agreement subprotocol requires only a supermajority of replicas, 
some replicas may fall behind other replicas. 
The state transfer subprotocol allows lagging replicas to catch up by 
transferring state from up-to-date replicas. 
To reduce the memory footprint, the checkpoint subprotocol is used to 
periodically garbage-collect state related to commands that have been executed 
in at least correct replicas.

\subsection{Hybrid Consensus}
\label{sec:bg:hybrid}

The BFT model allows malicious replicas to behave arbitrarily. 
Corrupt replicas can stop sending messages to one or more replicas, or send 
conflicting messages (equivocate) to different replicas with an intention to 
break safety. 
Preventing equivocation can reduce the number of replicas and size of quorums 
required to reach agreement
\cite{Chun:2007:A2M:SOSP:1294261.1294280,Levin:2009:TrInc:NSDI:1558977.1558978,Clement:2012:OLPNE:2332432.2332490}.
This is accomplished in the Hybrid fault model using a trusted component.
The algorithm hosted in the trusted component attests messages in such a way 
to prove their uniqueness. 
A simple monotonically increasing counter can be used for this 
purpose~\cite{Levin:2009:TrInc:NSDI:1558977.1558978}. 
By assigning a unique counter value per message and signing it, 
the trusted component ensures that the replica hosting it cannot send different 
messages with the same counter value. 
This property allows correct replicas to detect an attempt to send 
conflicting statements without requiring any additional communication 
mechanisms.
Thus, hybrid replicas can tolerate $f$ failures using only $2f+1$ replicas.

Figure~\ref{fig:bg:walkthrough} illustrates the agreement protocols of 
PBFT and MinBFT~\cite{Veronese:2013:Efficient:TC.2011.221}, 
a hybrid protocol, in a system tolerating $f=1$ failures. 
It can be observed that PBFT requires one additional replica and  
one additional communication phase compared to MinBFT.

\subsection{System Model}
\label{sec:bg:model}

We consider a system that consists of a set of nodes, called \textit{replicas} 
that communicate via message passing. 
These replicas implement a replicated service that receive commands from client 
and ensure that the same sequence of totally ordered commands are executed 
on the replicated state and responses are returned to the client. 
The goal of the consensus protocol is to ensure agreement on the replicated 
state among replicas withstanding a number of faulty servers.

Most consensus protocols offer a single fault model that the replicas use to 
commit the sequence of client commands and that the clients use to receive 
acknowledgement.
In contrast, \thesystem commits commands under two fault models and 
also lets clients choose the fault model for their response.
Thus, based on the assumptions, replicas may commit command sequences 
differently.
\thesystem provides the following guarantees:
\begin{itemize}[-]
    \item \textbf{Safety.}
    Any two replicas with correct but potentially different assumptions 
    commit the same sequence of client commands.
    \item \textbf{Liveness.}
    A command proposed by a replica will be eventually executed by 
every replica with a correct assumption.
\end{itemize}

\textbf{Fault Model.}
\thesystem supports two fault models: the Byzantine fault model and the 
hybrid fault model. 
In both the models, a replica is correct if it strictly follows the 
algorithm, otherwise it is faulty. 
In addition, faulty replicas can collude to harm the correct replicas.

To satisfy the hybrid model, we assume the existence of a trusted execution 
environment in each replica that hosts the protocol's trusted code. 
Despite some replicas being faulty, the trusted execution environment in 
each replica is assumed to be tamperproof and the code it executes 
strictly follows the algorithm. The trusted component can only fail by crashing. 
Note that this requirement is not necessary for BFT guarantees.

\textbf{Timing Model.}
We assume the partially synchronous timing model~\cite{Cachin:2011:RSDP}. 
Eventually, there exists a time during which correct replicas communicate 
synchronously and messages are timely. 
We, further, assume that the network can drop, reorder, and duplicate messages. 
To ensure reliable delivery of messages, we rely on generic retransmission 
techniques that use a buffer to store outgoing messages and retransmits 
them periodically.
Furthermore, we do not assume any bounds on processing and communication delays 
except that such delays do not grow indefinitely. 

\textbf{Cryptography.}
We assume that the adversary cannot break cryptographic computation such as 
hashes and signatures. 
In addition, the hashing algorithms are collision resistant.
Every replica is aware of other replicas' public keys.
Each replica can verify the messages they receive using the corresponding 
replica's public key.

\section{Flexible Quorums in Hybrid Model}\label{sec:minbft}

In this section, we revisit the quorum intersection in the Hybrid 
model, and introduce Flexible Hybrid Quorums, a technique that relaxes the 
majority quorum intersection requirement. 
With this technique, only quorums across views must intersect to ensure safety, 
while quorums within the same view need not intersect. 
Consequently, hybrid protocols using the flexible quorum technique can opt for 
using smaller non-majority quorums for agreement, in exchange for 
using much larger than majority quorums during view changes.
We perform our analysis in the context of MinBFT, a state-of-the-art hybrid 
protocol~\cite{Veronese:2013:Efficient:TC.2011.221}. 
Thus, we first overview MinBFT in Section~\ref{sec:minbft:bg}, and 
then introduce the flexible quorum technique to 
produce \textit{Flexible} MinBFT in Section~\ref{sec:minbft:flex}.

\subsection{Revisiting MinBFT}\label{sec:minbft:bg}

MinBFT~\cite{Veronese:2013:Efficient:TC.2011.221} is a hybrid 
fault-tolerant protocol that uses a trusted component to require 
only $N = 2f+1$ replicas 
to tolerate $f$ Byzantine faults. 
The trusted component prevents malicious replicas from equivocating to 
correct replicas, providing an efficient solution to the consensus 
problem under this hybrid model.

\subsubsection*{USIG Trusted Component}
The protocol uses a trusted component called the Universal Sequential Identifier 
Generator (USIG) that is present in each replica and provides two interfaces: 
one for signing and another for verifying messages. 
The USIG component assigns monotonically increasing counter values to messages 
and signs them.
The component provides the following properties: 
\begin{enumerate*}[label=(\roman*)]
    \item Uniqueness: no two messages are assigned the same identifier;
    \item Monotonicity: a message is never assigned an identifier smaller than 
    the previous one; and
    \item Sequentiality: the next counter value generated is always one 
    more than the last generated value.
\end{enumerate*}

To access its service, USIG provides two interfaces: 
\begin{itemize}[label=-]
    \item \texttt{CreateUI($m$)} creates a signed certificate $UI_i$ for  
message $m$ with the next value from the monotonic counter. The certificate is 
computed using the private key of the USIG instance $r$. 
    $UI_r = \langle ctr, H(m) \rangle_p$, where $ctr$ is the counter value and 
    $H(m)$ is the hash of the message.
    \item \texttt{VerifyUI($UI_i$, $m$)} uses the USIG instance $i$'s public key
and verifies whether the certificate $UI_i$ was computed for message $m$.

\end{itemize}

\begin{figure}
\begin{mdframed}
    A MinBFT replica executes the following protocol.
    \begin{enumerate}[label=(\arabic*)]
        \item \textbf{Prepare.} The primary assigns a sequence number to the 
client command and sends a \texttt{Prepare} message to all replicas. 
        \item \textbf{Commit.} Each replica receives the \texttt{Prepare} 
message and broadcasts a \texttt{Commit} message. 
        \begin{itemize}[label=-]
            \item A replica accepts a command if it collects 
            a commit certificate consisting of $f+1$ \texttt{Commit} 
            messages. 
            \item If a replica does not hear back from the primary in time, 
it will send a \texttt{ReqViewChange} message. 
            \item If a replica receives $f+1$ \texttt{ReqViewChange} messages,
it transitions to the next view and sends the \texttt{ViewChange} message.
        \end{itemize}
    \end{enumerate}
\end{mdframed}
\caption{MinBFT Normal Execution}
\end{figure}

\subsubsection*{MinBFT}
The MinBFT protocol proceeds in a sequence of views. The primary for each view 
is replica $r_i$ where $i = v \mod N$, $v$ is the view number and $N$ is the 
system size. The primary is responsible for handling client commands, 
assigning sequence number to those commands, and forwarding the commands to 
the replicas. The sequence numbers the primary assigns to commands is generated 
by the USIG instance within the primary. The replicas accept the 
command and execute it once they collect a commit certificate. 
A commit certificate indicates that a majority of replicas have observed the 
same message from the primary.

When the primary receives a client command $m$ with operation $o$, it assigns 
a sequence number to the command. 
The sequence number is the one generated by the USIG service. 
The primary $r_i$ sends the command in a message 
$\langle \texttt{Prepare}, v, r_i, m, UI_i \rangle$, where $UI_i$ contains the
unique sequence number and the signature obtained from the USIG module. 
Each replica $r_j$ in turn sends the 
$\langle \texttt{Commit}, v, r_j, r_i, m, UI_i, UI_j \rangle$ message to all 
other replicas. A client command is accepted at a replica if it receives
$f+1$ valid \texttt{Commit} messages, called a \textit{commit certificate}. 

Correct replicas only responds to the primary's \texttt{Prepare} message if 
the following conditions hold: $v$ is the current view number and the sender of 
the message is the primary of $v$; the USIG signature is valid; and that the 
messages are received in sequential order of the USIG counter value.
To prevent a faulty replica from executing the same operation twice, each 
replica maintains a $V_{req}$ to store the command identifier of the latest
operation executed for each client. The messages are always processed in 
the order of the USIG sequence number to prevent duplicity of operations 
and holes in the sequence number space. Replicas only execute an operation if 
it has not been executed already.

The view change protocol is triggered if the current primary fails to 
make timely progress. 
Replica sends a 
$\langle \texttt{ReqViewChange}, r_i, v, v' \rangle$ message to other replicas 
if it times out waiting for messages from the primary.
A replica moves into a new view if it receives $f+1$ \texttt{ReqViewChange} 
messages and consequently broadcasts a 
$\langle \texttt{ViewChange}, r_i, v', C_{l}, O, UI_i \rangle$ message, 
where $C_l$ is the last stable checkpoint certificate, 
$O$ is the set of generated messages since the last checkpoint. 
The new primary computes and sends a \\
$\langle \texttt{NewView}, r_i, v', V_{vc}, S, UI_i \rangle$, where $V_{vc}$ is 
the \textit{new view certificate} that contains the set of 
$f+1$ \texttt{ViewChange} 
messages used to construct the new view and $S$ is the set of prepared or 
committed commands since the last checkpoint. Replicas validate the received 
\texttt{NewView} message, 
update its sequence of operations to match $S$,
executes the pending operations, and 
starts accepting messages in the new view $v'$.

For conciseness, we defer the explanation of the checkpoint and state transfer 
procedures to the original paper~\cite{Veronese:2013:Efficient:TC.2011.221}.

\subsection{MinBFT with Flexible Quorums}\label{sec:minbft:flex}

In this section, we introduce the notion of Flexible Hybrid Quorums. 
First, we show that not all quorums need to intersect and consequently show 
that the system size $N$ need not be a function of $f$. 
Specifically, we show that only quorums of different kinds must intersect.
Thus, sizes of commit quorums $Q_c$ can be reduced at the 
cost of increasing the sizes of the view change quorums $Q_{vc}$.
We apply this technique to MinBFT and 
call the resulting protocol as \textit{Flexible} MinBFT.

MinBFT uses simple majority quorums for both the commit and the new view 
certificates. 
Thus, every quorum intersects with every other quorum. 
Consequently, commit quorums $Q_c$ intersect with other commit quorums. 
However, this is excessive. 
In MinBFT, the replica at the intersection of any two commit quorums, ensures 
that an operation is assigned only one $UI$ certificate. 
Note that the primary's USIG service already ensures that an $UI$ is assigned 
only once.
If the primary and the intersecting replica are malicious, the replica may 
still vote for the same operation at two different $UI$s. 
Correct replicas will handle this using the $V_{req}$ data structure. 
Since they process the messages in $UI$ order, they will observe that an 
operation has already been executed and not execute it again. 
This makes the intersection replica redundant.  
Thus, we relax the assumption that the different commit quorums intersect with 
each other. At the same time, to tolerate $f$ failures, the commit quorums 
should consist of more than $f$ replicas. Thus, we have that $|Q_c| > f$.

On the other hand, any view change quorum $Q_{vc}$ must intersect with any 
commit and view change quorums to ensure that the decisions made within a 
view are safely transitioned to future views. Hence, we have that 
$|Q_{vc}| + |Q_c| > N$.

The Flexible Hybrid Quorum requirement is captured by the following equations:
\begin{equation}\label{eq:minbft:flex1}
    |Q_c| > f
\end{equation}
\begin{equation}\label{eq:minbft:flex2}
    |Q_{vc}| + |Q_c| > N
\end{equation}

By setting $Q_c$ to the smallest possible value i.e. $|Q_c| = f+1$, we can 
observe that $|Q_{vc}|$ should equal $N - f$, to satisfy 
Equation~\ref{eq:minbft:flex2}.
Consequently, the system size $N$ need not be a function of $f$.

We applied the Flexible Hybrid Quorums to MinBFT. 
The resulting algorithm, \textit{Flexible} MinBFT remains the same 
except for the quorums they use. 
First, the commit quorums size is $Q_c = f+1$, but 
the variable $f$, the number of tolerated faults, is independent of $N$, 
the system size. 
Second, the size of view change quorums $Q_{vc}$ now equals $N-f$ instead of $f+1$.
The protocol does not require any other changes.

The safety and liveness guarantees provided by Flexible MinBFT are given below.
We only present the intuition and the related lemmas here. 
The complete proof is presented in the 
Appendix~\ref{sec:minbft:proof}.

\subsubsection*{Safety within a view.} 
This is ensured by the trusted subsystem and the commit quorums.
The trusted subsystem prevents equivocation, so a Byzantine replica 
cannot send conflicting proposals. 
Correct replicas will only vote on the proposed operation if the proposal 
is valid, and if it has not voted for the same operation before.
The following Lemma formalizes this notion. 

\begin{lemma}\label{lm:minbft:s1}
    In a view $v$, if a correct replica executes an operation $o$ with 
    sequence number $i$, no correct replica will execute $o$ with sequence 
    number $i' \neq i$.
\end{lemma}
    
\subsubsection*{Safety across views} 
    This is ensured by the trusted subsystem and the view change quorums. 
    The intersection of the commit and the view change quorum consists of at 
    least one replica. 
    Thanks to the trusted component, the replica in the intersection cannot 
    equivocate, and must reveal the correct sequence of operations executed, 
    as otherwise there will be holes that correct replicas can detect.
    Thus, a correct primary will gather the correct sequence and apply it in 
    the next view.

    \begin{lemma}\label{lm:minbft:s2}
        If a correct replica executes an operation $o$ with sequence number $i$ in 
        a view $v$, no correct replica will execute $o$ with sequence number 
        $i' \neq i$ in any view $v' > v$.
    \end{lemma}
    
\subsubsection*{Quorum availability.} 
    A non-faulty primary will always receive responses from a quorum of $f+1$ 
    replicas and this quorum will contain at least one honest replica. 
    \begin{lemma}\label{lm:minbft:l1}
        During a stable view, an operation requested by a correct client 
        completes.
    \end{lemma}

\begin{lemma}\label{lm:minbft:l2}
    A view $v$ eventually will be changed to a new view $v' > v$ if at least 
    $N-f$ correct replicas request its change.
\end{lemma}

\section{\thesystem}\label{sec:duobft}

\begin{figure}[t]
	\centering
	\begin{subfigure}{0.40\textwidth}
        \centering
		\includegraphics[width=\linewidth]{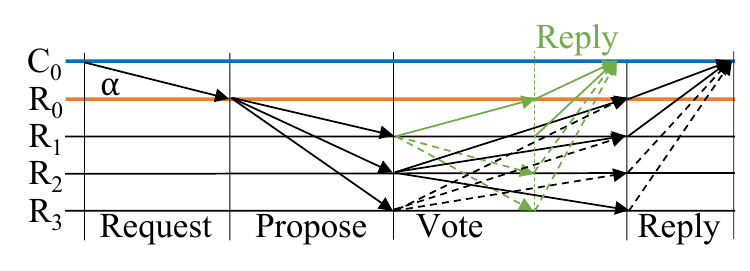}
		\caption{\thesystem}
        \label{fig:duobft:quadratic}
	\end{subfigure}
	\begin{subfigure}{0.45\textwidth}
        \centering
		\includegraphics[width=\linewidth]{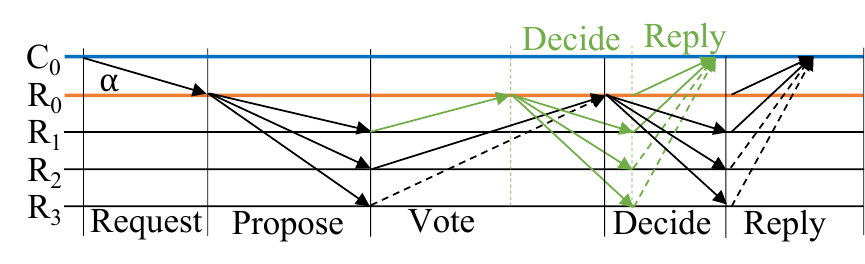}
		\caption{\thesystem with linear communication}
        \label{fig:duobft:linear}
	\end{subfigure}
	\caption{An example of normal execution steps in \thesystem with 
    quadratic and linear communication message complexities. The dotted arrows 
    represent non-quorum messages and the green arrows denote the Hybrid commit 
    path.}
	\label{fig:duobft}
\end{figure}

\thesystem is a BFT protocol that infuses a trusted component to provide 
better performance by leveraging the hybrid fault model while 
providing provisions for ensuring BFT resiliency for every client command. 
Client commands are committed under both hybrid and Byzantine fault  
assumptions.
This allows clients to take advantage of the quicker and cheaper hybrid 
commits and receive responses faster, improving system performance. 
However, since BFT commits also happen in tandem, Byzantine safety can be 
guaranteed despite trusted component compromises.

\thesystem requires $N = 3f+1$ replicas to tolerate up to $f$ failures.
To tolerate hybrid failures, we assume a non-Byzantine trusted execution 
environment in all replicas including malicious ones, 
but tolerating Byzantine failures require no such assumptions.
The trusted execution environment in each replica will host the USIG service 
(described in Section~\ref{sec:minbft:bg}) for certifying the messages shared 
by the replicas.
For ease of exposition, we use $f_B$ to denote Byzantine failures and 
$f_H$ to denote hybrid failures.
However, in the context of \thesystem, $f_B = f_H = f$. 
Replicas collect both the hybrid and BFT quorums, commit client commands 
under the respective fault models, and respond back to the client.
The client can specify its response fault model along with the command sent 
to the replicas.
Furthermore, we assume a mechanism for clients to obtain responses under 
both the fault models either via a long-lived client-replica connection or 
a separate request-response mechanism.

\thesystem is composed of an agreement protocol that collects 
two kinds of quorum votes, and a view-change protocol with a single quorum type. 
\thesystem's ability to have one view-change protocol for both the 
fault models is made possible by the Flexible Hybrid Quorums model 
(Section~\ref{sec:minbft:flex}).
Recall that, to tolerate $f$ failures, a traditional hybrid protocol requires
$N_H=2f+1$ replicas, while a BFT protocol requires $N_B=3f+1$ replicas.
The quorums are thus $f+1$ and $2f+1$ respectively, including for view changes.
In contrast, the Flexible Hybrid Quorums enable having more than $N_H$ replicas 
without changing $f$, since $N$ is not a function of $f$.
This allows running a hybrid protocol in a system with $N=3f+1$ replicas 
by adapting the quorum sizes. 
We use $f+1$ quorums for agreement subprotocol while $2f+1$ quorums for 
view-change protocol. 
Since the view-change procedure for both the BFT and hybrid fault models 
are the same, \thesystem is able to have a unified view-change protocol using 
the matching view-change quorums in both fault models.

In the rest of this section, we will present the detailed description of 
\thesystem, overview its properties, and discuss its guarantees.
Similar to other recent works such as 
Casper~\cite{Buterin:2017:Casper:1710.09437},
Hotstuff~\cite{Yin:2019:HotStuff:3293611.3331591}, 
and Flexible BFT~\cite{Malkhi:2019:FlexibleBFT:CCS:3319535.3354225}, 
we explain \thesystem in terms of a Blockchain protocol
where the votes are pipelined.
First, we present the terminologies.

\begin{figure*}
    \begin{mdframed}
    A \thesystem replica executes the following protocol:
    
    \textbf{Normal Protocol}
        \begin{enumerate}[label=(\arabic*)]
            \item \textbf{Propose.} The primary $P$ creates a block and sends it 
    in a $\langle \texttt{Propose}, v, P, B_k, s_P \rangle$ message to all 
    replicas. 
    It attests the block with a USIG certificate $UI$. 
    The primary collects votes for the blocks from the replicas. 
    The primary sends the next block when it receives a quorum certificate for 
    its previous block.
            \item \textbf{Vote.} A replica $R$ receives the  
    $\langle \texttt{Propose}, v, P, B_k, s_P \rangle$ message from the primary, 
    validates if it extends the last proposed block, and votes for it. 
    The vote is sent in a $\langle \texttt{Vote}, v, R, B_k, s_P, s_R \rangle$ 
    message to other replicas containing a $UI$ certificate.
    
    In addition, Replica $R$ records the following information for a block:
    \begin{itemize}[label=-]
        \item $q_v(B_k)$: The votes received for block $B_k$ by replica $R$ 
from any other replica in view $v$.
        \item \textbf{$\mathcal{C}_{H}(B_k)$}: 
A set of $f+1$ valid votes form a hybrid quorum certificate 
for block $B_k$. 
        \item \textbf{$\mathcal{C}_{B}(B_k)$}:
A set of $2f+1$ valid votes form a BFT quorum certificate for block $B_k$. 
    \end{itemize}
        \end{enumerate}

    \textbf{View Change Protocol}

    \begin{enumerate}[label=(\arabic*)]
        \item \textbf{View Change Request}
    A replica requests a view change if it does not receive proposals from the 
    replicas in a timely manner, or if it observes equivocating blocks either 
    via the proposal or the vote messages.
    \begin{itemize}[label=-]
        \item Replica $R$ sends a 
$\langle \texttt{ReqViewChange}, R, v, v' \rangle$ to request a view change 
from $v$ to $v' = v+1$.
        \item A replica that receives $f+1$ \texttt{ReqViewChange} messages 
transitions to the new view and multicasts the \texttt{ViewChange} message to 
other replicas.
        \item A replica that receives $f+1$ \texttt{ViewChange} message also
transitions to the new view and sends its \texttt{ViewChange} message to other 
replicas. The View Change message consists of all the blocks that the replicas 
have a quorum certificate for.
    \end{itemize}
        \item \textbf{New View.} 
    The new primary $P'$ collects $2f+1$ \texttt{ViewChange} messages and 
    computes the sequence of blocks in the new view $v'$. It sends a 
    $\langle \texttt{NewView} \rangle$ message to all replicas.
        \item \textbf{New View Install.} 
    A replica that receives the new primary $P'$'s \texttt{NewView} messages, 
    validates it, and installs $S$, the block sequence in the new view.
    \end{enumerate}
    
    \end{mdframed}
    \caption{\thesystem Protocol Execution}
    \label{fig:duobft:protocol}
\end{figure*}

\subsection{Preliminaries}

\subsubsection{Blockchain} As presented in the previous sections, in classical 
BFT protocols, the agreement happens on a sequence of client commands. 
In a Blockchain protocol, the agreement happens on a chain of blocks, where 
each \emph{block} contains one or more client commands. 
Each block has a reference to a predecessor block except the first block 
in the chain, also called the genesis block, which has no predecessors. 
Every block has a height parameter that indicates its position from 
the genesis block. 
A block $B_k$ at height $k$ has the following format: 
$B_k := (b_k, h_{k-1})$ where $b_k$ refers to the block's 
value and $h_{k-1} = H(B_{k-1})$, the hash of the predecessor block in the 
chain. 
For the genesis block, the predecessor hash is null, thus 
$B_1 := (b_1, \bot)$. 
Note that only the genesis block can have a null 
predecessor hash; other block must specify a valid hash. 

\subsubsection{Block Prefix and Equivocation} 
Let $S$ be a sequence of blocks in increasing height order.
The prefix of a Block $B_k$ at height $k$ in sequence $S$, 
denoted $prefix(S, k)$, is the prefix of the sequence $S$ containing 
the first $k$ blocks from the genesis block.
Equivocation happens when the sequence of 
blocks diverges. 
Given two blocks $B_i$ and $B_j$, we say those blocks diverge 
when $B_i$ is not the ancestor of $B_j$ or vice versa. 

\subsubsection{Block Certificates} 
Replicas vote on the blocks by signing on the hash of the block $H(B_k)$. 
A quorum of these votes form a quorum certificate. 
In \thesystem, replicas collect two kinds of certificates for a block: 
a Hybrid quorum certificate \textbf{$\mathcal{C}_{H}(B_k)$} and a 
BFT quorum certificate \textbf{$\mathcal{C}_{B}(B_k)$}.
The size of the quorums are discussed below.

\subsection{Protocol}

The \thesystem protocol proceeds in a view by view fashion. 
The primary of each view is decided using the formula $v \mod N$ 
i.e. primary roles are assigned to replicas in round-robin order. 
The primary of the view is responsible for proposing blocks that other replicas 
vote on. 
Figure~\ref{fig:duobft:protocol} presents a concise algorithm description.

At a high level, \thesystem works as follows. 
The primary proposes a block to replicas. 
Replicas vote on the block if it is safe to do so.
A quorum of such votes on a block make a quorum certificate.
After collecting a quorum certificate for a block, the primary moves on to 
propose the next block extending the previous one.
We will discuss how commit decisions on blocks are made in 
Section~\ref{sec:duobft:commitrule}.
Replicas use the view change protocol to install a new view if they are unable 
to make progress in the current view. 
The view change protocol begin only if $f+1$ replicas request a view change.

\subsubsection{Normal protocol} The normal protocol is executed when a view is 
stable. 
In a stable view $v$, the primary and $N-f$ replicas behave 
correctly and exchange messages in a timely manner. 
The primary creates a new block $B_k$ that extends the highest block in the 
chain it is aware of and signs the block using its USIG component.
The primary $P$ sends the block to the replicas in a 
$\langle \texttt{Propose}, v, P, B_k, s \rangle$ message, 
where $v$ is the current view and $s$ is the USIG certificate $UI$. 

A replica $R$ that receives the \texttt{Propose} message for block $B_k$, votes 
on the block if it extends the previously proposed block in the view. 
Similar to MinBFT, \thesystem replicas only process blocks in increasing height 
order.
When blocks are received out of order, replicas wait to receive the 
predecessor blocks to $B_k$ and validates the blocks, casts its vote in the 
height order.
Replica $R$ creates a $UI$ certificate and sends its vote in a 
$\langle \texttt{Vote}, v, R, B_k, s_P, s_R \rangle$ message to other replicas.
The votes collected by a replica can form two kinds of quorum 
certificates for the block $B_k$. 
Every replica records the following information for a block:
\begin{itemize}[label=-]
    \item \textbf{$\mathcal{C}_{H}(B_k)$}: 
A set of $f+1$ votes with valid $UI$ certificates form a hybrid quorum certificate 
for block $B_k$.
    \item \textbf{$\mathcal{C}_{B}(B_k)$}:
A set of $2f+1$ votes with valid certificated form a BFT quorum certificate 
for block $B_k$. 
\end{itemize}

\subsubsection{View Change} 
If a replica detects equivocation or lack of progress by the primary, 
it will start the view change procedure to move from the current view $v$ to 
the next stable view $v'$. 
A replica requests a view change by sending a 
$\langle \texttt{ReqViewChange}, R, v, v' \rangle$ message to 
other replicas. 
When a replica receives at least $f+1$ \\ \texttt{ReqViewChange} messages, it 
starts the view transition and sends a 
$\langle \texttt{ViewChange}, R, v', O, UI_i \rangle$  message, where 
$O$ contains the sequence of blocks for which $R$ has collected any 
quorum certificate. 
The primary $P'$ of the new view $v'$ will collect $2f+1$ valid 
\texttt{ViewChange} messages to form the new view certificate. 
$P'$ will use this certificate to compute the set of blocks $S$ for which
quorum certificates exist. 
$P'$ sends a 
$\langle \texttt{NewView}, r_i, v', V_{vc}, S, UI_i \rangle$, where $V_{vc}$ is 
the \textit{new view certificate} that contains the set of 
$f+1$ \texttt{ViewChange} 
messages used to construct the new view and $S$ is the set of prepared or 
committed requests. 
Replicas verify the validity of the $S$ by performing the same computation as 
the new primary using the new view certificate. 
Then, replicas adjust their local state according to $S$ and start voting in 
the new view $v'$.

The new view computation performed by the replicas and the primary is similar 
to signature-based PBFT's view change mechanism~\cite{Castro:2001:PBFT:Thesis}. 
We omit the details for conciseness.
Figure~\ref{fig:duobft:protocol} presents a brief overview of view-change steps.

\subsection{Commit Rules}
\label{sec:duobft:commitrule}

\begin{figure}
    \begin{mdframed}
        \begin{enumerate}
            \item \textbf{Hybrid Commit Rule:} A replica commits a block $B_k$ 
under the Hybrid Commit Rule iff it collects hybrid quorum certificate for 
block $B_k$ i.e. $\mathcal{C}_{H}(B_k)$.
            \item \textbf{BFT Commit Rule:} A replica commits a block $B_k$ 
under the BFT Commit Rule iff it collects BFT quorum certificates for 
the blocks $B_k$ and $B_{k+1}$ i.e. $\mathcal{C}_{B}(B_k)$ and 
$\mathcal{C}_{B}(B_{k+1})$, and $B_{k+1}$ extends $B_k$.
        \end{enumerate}
    \end{mdframed}
    \caption{\thesystem Commit Rules}
\end{figure}

In \thesystem, replicas commit blocks under two different fault models 
by reusing the same set of vote messages for a given block. 
The protocol enforces a set of commit rules that uses the vote 
messages and chain state to decide when to commit the blocks.

\subsubsection{Hybrid Commit Rule}
A replica can commit a block under the Hybrid fault model 
when it receives at least $f+1$ votes, called the hybrid quorum certificate,
from replicas. 
Under the Hybrid Commit Rule, the protocol provides the same safety guarantees 
as Flexible MinBFT.

\subsubsection{BFT Commit Rule} 
A replica can commit blocks under the Byzantine fault model 
when it receives at least $2f+1$ votes, called the BFT quorum certificate, 
for the block $B_k$ and its parent block $B_{k-1}$. 
Under the BFT Commit Rule, the protocol provides the same safety guarantees as 
PBFT.

\subsection{Proof}

\begin{lemma}\label{lm:duobft:s1}
    If a replica commits a block $B_l$ in a view $v$, then no replica
    with the same assumptions will commit $B'_{l}$ that does not equal $B_l$
    in view $v$.
\end{lemma}
\begin{proof}
    \textbf{BFT Commit Rule:}
    We prove by contradiction.
    Say a replica commits block $B_l$. It will have $q_c$ votes for $B_l$ 
    and its immediate successor.  
    Suppose another replica commits block $B'_l$ then it will have $q_c$ votes 
    for $B'_l$ and its immediate successor. 
    However, the intersection of two $q_c$ quorums will have at least one 
    correct replica that will not vote for two blocks at the same height. This 
    is a contradiction.

    \textbf{Hybrid Commit Rule:}
    We prove by contradiction.
    Say a replica commits block $B_l$. It will have $f+1$ USIG votes for $B_l$.  
    Suppose another replica commits block $B'_l$ then it will have $f+1$ USIG
    votes for $B'_l$. 
    However, a primary cannot sign two messages with the same USIG identifier.
    Thus, there is no way there can exist two blocks at the same height $l$.
    This is a contradiction.

    Thus, it is not possible for any two replicas to commit different blocks
    at the same height in view $v$.
\end{proof}

\begin{lemma}\label{lm:duobft:s2}
    If a replica commits a block $B_l$ in a view $v$, no replica with the 
    same assumptions will commit block $B'_{l}$ that does not equal $B_l$ at 
    the same height $l$ in any view $v' > v$.
\end{lemma}
\begin{proof}
    \textbf{BFT Commit Rule:}
    We prove by contradiction.
    A replica commits block $B_l$ in view $v$ then it will have $q_c$ votes 
    for $B_l$ and its immediate successor.  
    Suppose another replica commits block $B'_l$ in view $v' > v$ then it 
    should have $q_c$ votes for $B'_l$ and its immediate successor. 
    
    Since $B_l$ is committed in view $v$, there should exist quorum 
    certificates for blocks $B_l$ and $B_{l+1}$. In the new view $v' = v+1$, 
    the \texttt{NewView} message sent by the new primary includes a 
    \texttt{NewView} certificate 
    with $q_r$ \texttt{ViewChange} messages that contains at least one 
    correct replica. 
    The \texttt{ViewChange} message from the correct replica 
    will have the correct certificate for $B_l$ and $B_{l+1}$ 
    Thus, the new primary must enforce blocks $B_l$ and $B_{l+1}$ in the new 
    view $v+1$. It will receive votes only for $B_l$ in the new view $v+1$. 
    For $B'_l$ to be committed, $q_c$ replicas must votes for it, which cannot 
    happen since there is at least one correct replica in the intersection of 
    $q_r$ and $q_c$ that received the NewView message with correct certificates.
    This is a contradiction. Thus, $B'_l$ cannot have been committed in $v+1$. 

    \textbf{Hybrid Commit Rule:}
    We prove by contradiction.
    A replica commits block $B_l$ in view $v$ then it will have $f+1$ votes 
    for $B_l$.  
    Suppose another replica commits block $B'_l$ in view $v' > v$ then it 
    should have $f+1$ votes for $B'_l$. 
    
    Since $B_l$ is committed in view $v$, there should exist a USIG quorum 
    certificate for block $B_l$. In the new view $v' = v+1$, 
    the \texttt{NewView} message sent by the new primary includes a new view certificate 
    with $N-f$ \texttt{ViewChange} messages that contains at least one correct replica. 
    The correct replica's \texttt{ViewChange} message will have the correct 
    quorum certificate for $B_l$. 
    It might happen that the new primary might remove some block entries 
    from the \texttt{ViewChange} message, but this will be detected as the 
    USIG-signed \texttt{NewView} message will reveal the holes in the message log 
    (See Lemma~\ref{lm:minbft:l2} for additional details.)
    Thus, the new primary must enforce blocks $B_l$ and $B_{l+1}$ in the new 
    view $v+1$. It will receive votes only for $B_l$ in the new view $v+1$. 
    For $B'_l$ to be committed, $q_c$ replicas must votes for it, which cannot 
    happen since there is at least one correct replica in the intersection of 
    $q_r$ and $q_c$ that received the NewView message with correct certificates.
    This is a contradiction. Thus, $B'_l$ cannot have been committed in $v+1$.

    For both commit rules above, the case for arbitrary $v' > v$ where 
    $v' = v+k$ will fall under 
    the case of $v+1$, since at each view transition, the information from 
    one view is propagated to the next view.
\end{proof}

\begin{theorem}\label{th:duobft:s1}
    Any two replicas with the same commit rule commit the same sequence of 
    blocks in the same order. 
\end{theorem}
\begin{proof}

    To elaborate on the theorem, if a replica following a commit rule commits 
    the sequence of blocks $S = \langle B_1 ... B_i \rangle$, then another 
    replica that follows the same commit rule will commit the same sequence of 
    blocks $S$ or a prefix of it. 
    We use $prefix(S, i)$ to represent the first $i$ blocks of the sequence $S$.
    We use the $\bullet$ operator to concatenate any two sequences.

    Assume the theorem is false i.e.  there should exist two sequences $S$ and 
    $S'$ committed by two replicas that is not a prefix of each other. Assume 
    the sequences conflict at $i$, such that 
    $prefix(S, i) = prefix(S', i-1) \bullet \langle B_i \rangle$ and
    $prefix(S', i) = prefix(S, i-1) \bullet \langle B'_i \rangle$.    
    Precisely, there exists two blocks $B_i$ and $B'_i$ 
    at the same height $i$ committed by two different replicas with the same 
    commit rule. Assume that block $B_i$ was committed in view $v$ and 
    block $B'_i$ was committed in view $v'$. 
    If $v = v'$, then this will contradict Lemma~\ref{lm:duobft:s1}. 
    If $v' > v'$, then this will contradict Lemma~\ref{lm:duobft:s2}.
    Hence, the theorem must hold.
\end{proof}

\begin{lemma}\label{lm:duobft:l1}
    During a stable view, a proposed block is committed by a replica.
\end{lemma}
\begin{proof}
    In a stable view, the correct primary will propose blocks in a timely 
    fashion. If the primary is hybrid, then it will generate an 
    $UI = \langle i, H(b) \rangle_p$ for the block. Correct replicas that 
    receive the proposal will vote for it. Replicas that are hybrid 
    will generate an $UI$ for their votes. Since there are at most $f$ 
    faulty replicas, they will remain $N - f$ correct ones. For a hybrid quorum,
    at least $f+1$ of these $N - f$ replicas will reply on time. Similarly, 
    for a BFT quorum $N-f = 2f+1$ replicas will reply on time. Thus, a replica 
    will receive the votes on time and will commit the block using their commit 
    rule.
\end{proof}

\begin{lemma}\label{lm:duobft:l2}
    A view $v$ will eventually transition to a new view $v' > v$ 
    if at least $N-f$ replicas request for it.
\end{lemma}
\begin{proof}

    A replica $R$ can request a view change by sending a 
    $\langle \texttt{ReqViewChange}, R, v, v' \rangle$
    message. The view change mechanism is triggered when replicas 
    receive $f+1$ ReqViewChange messages for the same view. Assume that 
    replicas collect $f+1$ messages for transitioning from $v$ to $v+1$. The 
    primary for the new view is $(v+1) \mod N$ by definition. Consider the 
    two cases:
    \begin{enumerate}
        \item \textit{the new view is stable}: correct replicas will receive the \\
\texttt{ReqViewChange} messages. Consequently, correct replicas that 
receive at least $f+1$ \texttt{ReqViewChange} messages will enter the new 
view $v'$ and send a \texttt{ViewChange} message to all replicas. The primary 
$p$, being stable, for view $v+1$ will send a valid \texttt{NewView} message 
in time. Thus, correct replica that receive the message will transition to 
new view $v' = v+1$.
        \item \textit{the new view is not stable}: We consider two cases:
        \begin{enumerate}
            \item \textit{the primary $p$ is faulty and does not send the 
\texttt{NewView} message in time, or $p$ is faulty and sends an invalid 
\texttt{NewView} message}, or $p$ is not faulty but the network delays $p$'s 
message indefinitely. In all these cases, the timer on other correct replicas 
that sent the \texttt{ViewChange} message will expire waiting for the NewView 
message. These replicas will trigger another view change to view $v+2$.
            \item \textit{the primary $p$ is faulty and sends the NewView 
message to only a quorum $Q'$ of $q_{vc}$ replicas but less than $q_{vc}$ 
replicas are correct, or $p$ is correct but there are communication delays.} 
The replicas in quorum $Q'$ may enter the new view and process requests in 
time. However, the correct replicas that does not receive the NewView message 
will timeout and request change to view $v+2$. However, there will be less than
$f+1$ replicas, so a successful view change trigger will not happen. If the 
faulty replicas deviate from the algorithm, other correct replicas will join 
to change the view. 
        \end{enumerate}
    \end{enumerate}
\end{proof}

\begin{theorem}
    A proposed block is eventually committed by replicas with correct commit rules.
\end{theorem}
\begin{proof}
    When the view is stable, Lemma~\ref{lm:duobft:l1} shows that the proposed block is 
    committed by the replicas. When the view is not stable and the replica timers 
    expire properly, $f+1$ replicas will request a view change. By 
    Lemma~\ref{lm:duobft:l2}, a new view $v'$ will be installed.

    However, if less than $f+1$ replicas request the view change, then the 
    remaining replicas that do not request the view change will follow the 
    protocol properly. Thus, the system will stay in view $v$ and the replicas 
    will continue to commit blocks in the view. When proposals are not 
    committed in time or when more than $f$ replicas request a view change, 
    then all correct replicas will request a view change and it will be 
    processed as in Lemma~\ref{lm:duobft:l1}.

    Even after a view change, the new view $v'$ may not necessarily be stable. 
    If the new primary deviates from the algorithm or does not process messages 
    in time, this will cause correct replicas to request another view change 
    and move to the next view. Since there can only be at most $f$ faulty 
    replicas, after at most $f+1$ view changes, a stable view will be installed.
    Furthermore, if the faulty primary follows the algorithm enough such that 
    a view change cannot be triggered, by Lemma~\ref{lm:duobft:l1}, replicas 
    will continue to commit the blocks.
    
\end{proof}

\subsection{Optimizations}

\subsubsection{Reducing Message Complexity.}
In the description of \thesystem presented in the previous section, 
the replicas multicast their  votes to all other replicas incurring an 
$O(N^2)$ message complexity in the common case. 
This complexity can be reduced to $O(N)$ by modifying the replicas to send 
their votes only to the primary and enabling the primary to collect the 
votes and share the quorum certificate with other replicas. 
This technique enables the use of 
threshold signatures schemes~\cite{stathakopoulous2017threshold} to reduce 
the size of outgoing messages from the primary and to reduce the 
verification compute overheads when $N$ is large.
Many existing protocols use this technique
\cite{Gutea:2019:SBFT:DSN.2019.00063}.
An illustration of this optimization is shown in Figure~\ref{fig:duobft}. 
While the aggregation increases the number of communication steps of the 
protocol, as we will show in Section~\ref{sec:eval}, 
reducing the complexity helps reduce the latency of the protocol at large 
system sizes ($N \geq 49$).

\subsubsection{MultiChain-\thesystem}

Chain-based protocols including \thesystem do not support out-of-order 
processing and thus exhibit poor throughput compared to protocols that 
support out-of-order processing (e.g. PBFT). 
Since the protocol phases are pipelined in chain-based protocols, 
the votes for the previous block must be available before sending 
the next block. 
Therefore, the throughput of such protocols is dependent on the network message 
delays that prominently determine how quickly replicas can collect 
a quorum of votes.
Hence, these protocols perform poorly in wide-area deployments where latencies 
between regions are large~\cite{Gupta:2020:ResilientDB:3380750.3380757}.
On the other hand, protocols such as PBFT can send propose \emph{blocks} 
simultaneously and collect multiple phases of votes for each of those blocks 
and provide higher throughput and lower latency. 

Despite, chain-based protocols are efficient in terms of the number 
of messages exchanged per block because they pipeline their votes. 
For \thesystem, this means that collecting different kinds of votes is 
possible without increasing the number of message types and messages 
exchanged to commit per block.
To compensate for the lost throughput, we propose running multiple instances of 
\thesystem concurrently to facilitate collecting votes for multiple blocks 
at the same time.
Various techniques to run multiple instances of a BFT protocol and coordinate 
ordering among those instances have been proposed in the past
\cite{Eischer:2017:OMADA:EDCC.2017.15,Voron:2019:Dispel:arxiv:1912.10367,Gupta:2021:RCC:ICDE,Stathakopoulou:2019:MirBFT:arxiv:1906.05552,Behl:2017:Hybster:Eurosys:3064176.3064213}. 
We adopt a recent multi-primary paradigm RCC~\cite{Gupta:2021:RCC:ICDE} and 
modify it slightly to run multiple instances with the same primary.
Our choice of RCC was due to the fact that the approach does not require 
changing the underlying protocol unlike COP~\cite{Eischer:2017:OMADA:EDCC.2017.15}. 
In RCC, each replica is primary for an instance of the Byzantine Agreement 
protocol and commits client commands in rounds, where in each round one 
command per replica is committed. 
Once all replicas have committed commands in a given round, they are executed 
in a pre-determined order.
Since our intent is to improve \thesystem's performance with a single 
primary, we simply assign the same replica to be the primary of multiple 
instances. 
We call this variant of \thesystem as MultiChain\thesystem or 
simply \themcsystem. 

The performance of the \themcsystem protocol now depends on the number 
of concurrent instances.
In our experiments, we manually fixed the number of instances depending on 
the system size. 
Typically, the number of instances was between 4 for large systems ($N=97$)
and 40 for smaller systems ($N=25$).
However, note that prior works~\cite{Voron:2019:Dispel:arxiv:1912.10367} have 
investigated the idea of automatically tuning the number of concurrent 
instances based on the available network and compute resources.
How those ideas integrate with RCC is beyond the scope of this paper.


\section{Discussion}\label{sec:discussion}

\textbf{Implications of Trusted Environment Compromises.}
As mentioned in the introduction, trusted execution environments are 
increasingly being scrutinized for security vulnerabilities.
In the hybrid fault model, the compromise of the trusted component is enough 
to break the safety of the protocol. 
However, \thesystem holds safety in such cases via the BFT commit rule. 
If the trusted component is compromised, per our assumption, this can only 
affect at most $f$ replicas. 
Thus, the remaining $2f+1$ replicas will follow the algorithm correctly.
While the hybrid quorum certificates can become invalid, recall that 
replicas also collect BFT quorum certificates in tandem.
Thus, safety is still preserved for the  sequence of blocks that have collected 
the BFT quorum certificates.

\textbf{Comparison to FlexibleBFT.}
We now highlight the differences of our protocol from 
Flexible BFT~\cite{Malkhi:2019:FlexibleBFT:CCS:3319535.3354225}, 
a recent protocol that provides diverse learner assumptions.
The first important distinction is that Flexible BFT provides the a-b-c fault 
model in addition to the BFT model. 
The replicas under the a-b-c model are allowed to attack the safety of the 
system, but when they aren't able to attach safety, they will ensure liveness. 
However, the implication of using this fault model is that Flexible BFT 
quorums are much larger than our flexible hybrid quorums.
In Flexible BFT, the commit quorums used by the client $q_c$ should be at 
least as large as the view change quorum $q_r$ used by the replicas, 
i.e $q_c \geq q_r$. 
In contrast, \thesystem uses hybrid commit quorums that are 
smaller than the view change quorums, and the BFT commit quorums are as large 
as the view change quorums. 
That is, in \thesystem, $q_c \leq q_{vc}$.

Furthermore, Flexible BFT uses the synchrony timing model as a means to 
provide commits using simple majority quorums, which are smaller than Byzantine 
quorums. The protocol also tolerate $< 1/2$ failures. 
On the other hand, \thesystem tolerates only $1/3$ failures, but under 
the hybrid model, its commit quorum sizes are really efficient, only a little 
over $1/3$ replicas. 
Thus, Flexible BFT uses the timing model to reduce 
quorum sizes, while \thesystem uses the trusted component to achieve a very 
similar purpose.  
Furthermore, partially synchrony model enables ``network-speed'' replicas 
those that do not need lock-step executions unlike in the synchrony model. 
Thus, assumptions such as globally synchronized clocks are not 
required in our case.


\section{Evaluation}\label{sec:eval}

\begin{figure*}[!h]
	\centering
	\begin{subfigure}[b]{\textwidth}
		\centering
		\includegraphics[width=\textwidth]{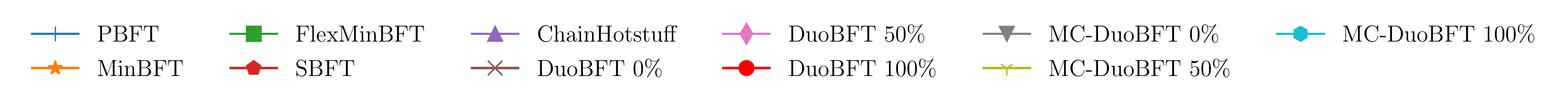}
	\end{subfigure}
	\begin{minipage}{.49\textwidth}
			\begin{subfigure}[b]{0.49\columnwidth}
				\centering
				\includegraphics[width=\textwidth]{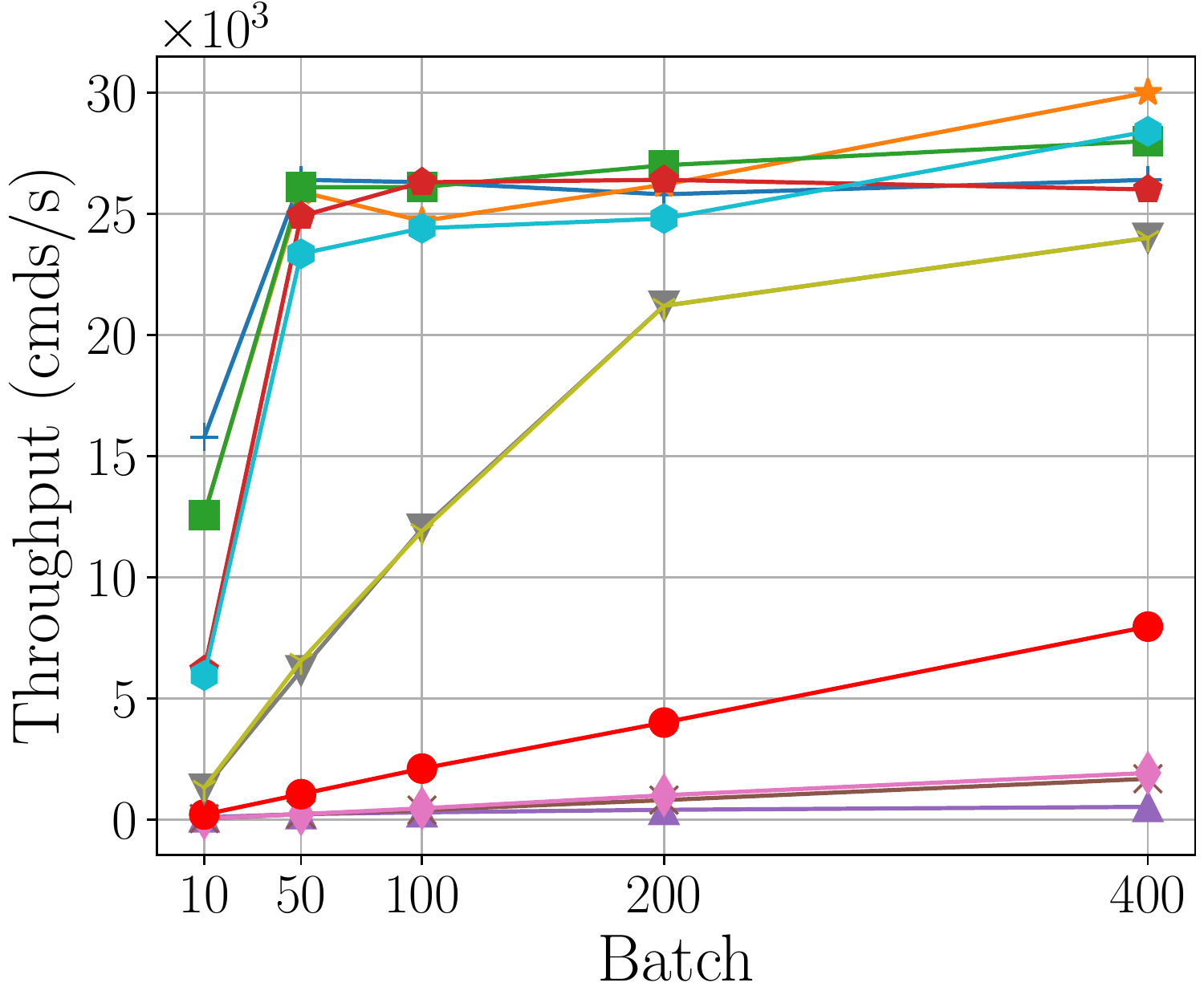}
			\end{subfigure}
			\begin{subfigure}[b]{0.49\columnwidth}
				\centering
				\includegraphics[width=\textwidth]{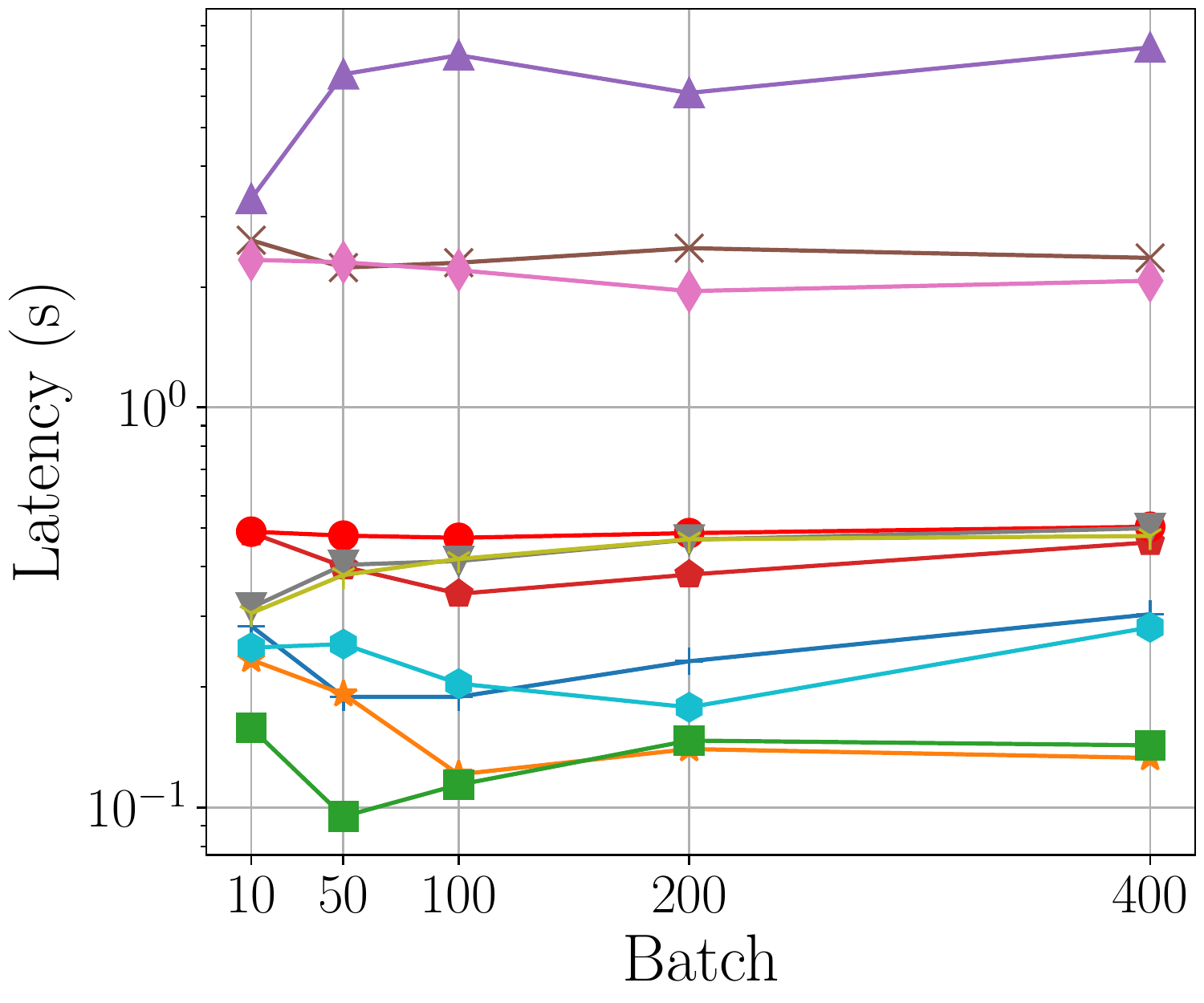}
			\end{subfigure}
			\caption{Performance versus Batch Size (N=49).}
			\label{fig:batching-plot}
	\end{minipage}
	\begin{minipage}{.49\textwidth}
			\centering
			\begin{subfigure}[b]{0.49\columnwidth}
				\centering
				\includegraphics[width=\textwidth]{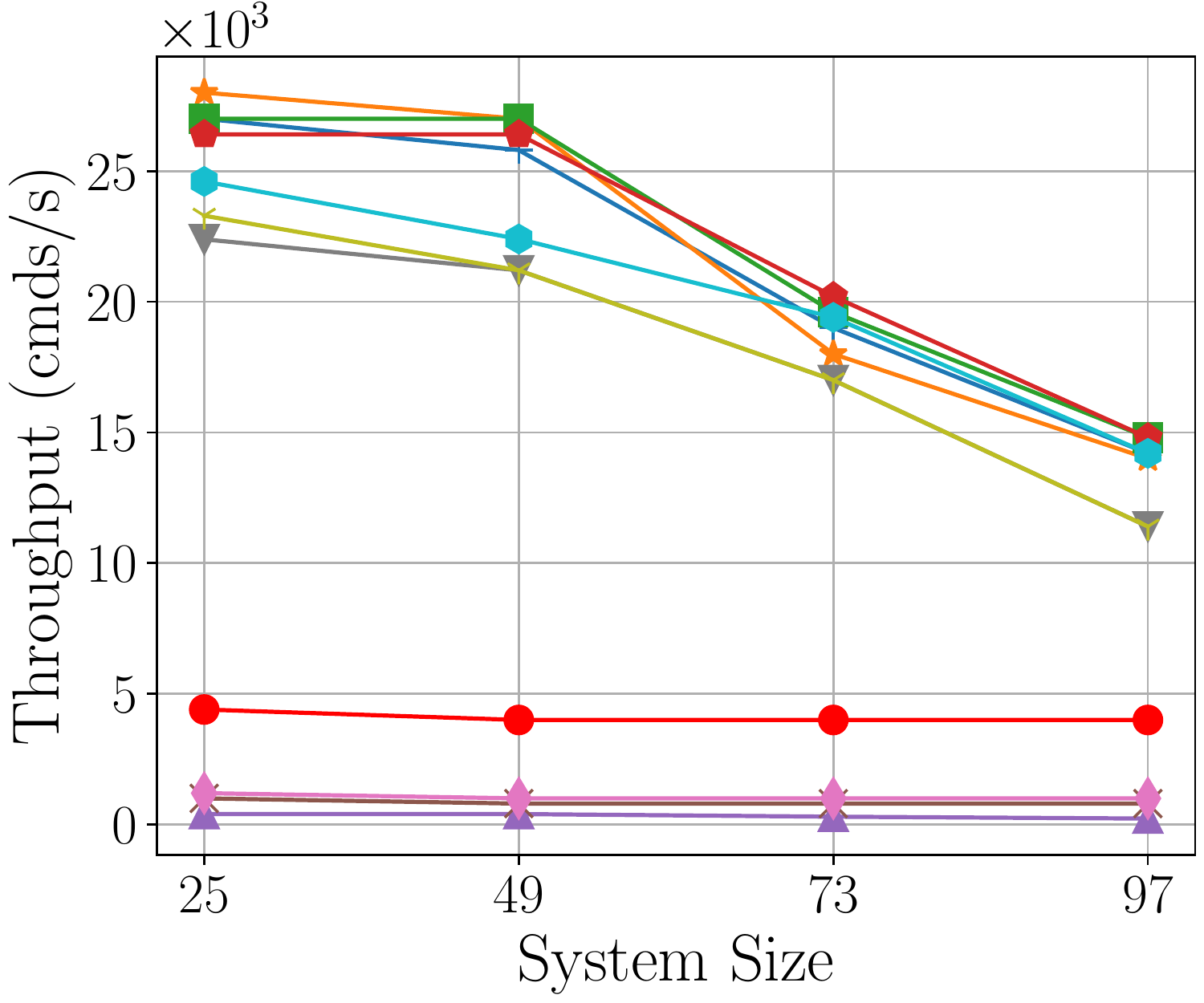}
			\end{subfigure}
			\begin{subfigure}[b]{0.49\columnwidth}
				\centering
				\includegraphics[width=\textwidth]{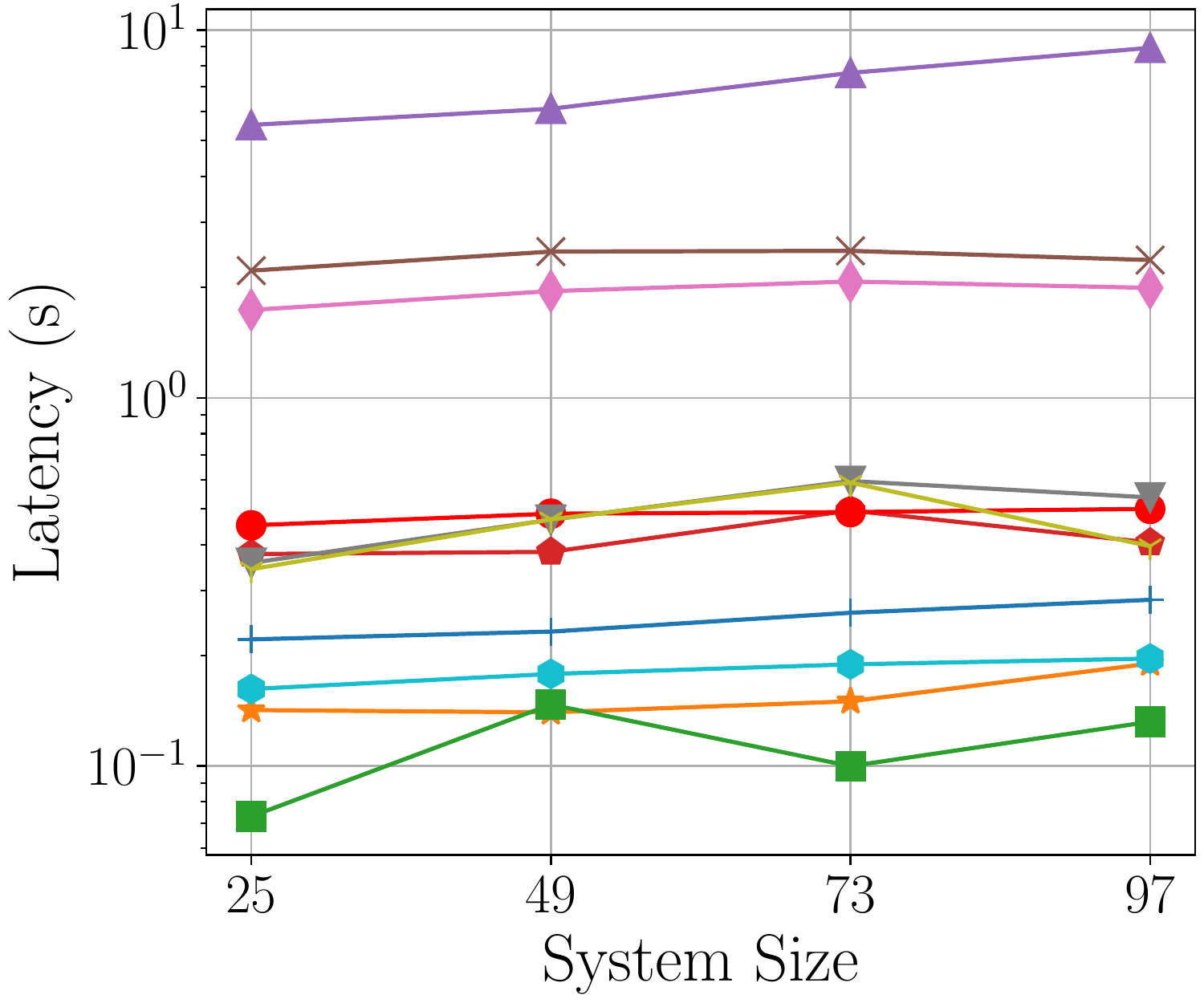}
			\end{subfigure}
			\caption{Performance versus System Size.}
			\label{fig:scalability-plot}
	\end{minipage}
\end{figure*}

In this section, we evaluate Flexible MinBFT, \thesystem, and \themcsystem 
alongside state-of-the-art protocols to answer the following questions:
\begin{enumerate}[label=(\arabic*)]
	\item What is the impact of batching on protocol performance?
    \item How does scale affect protocol performance?
    \item How well do the protocols cope with replica failures?
    \item Does \thesystem integrate with the recent multi-primary paradigms?
\end{enumerate}

Throughout our evaluation of \thesystem, we also measured the overhead of 
committing commands under two different fault models within the same protocol.

\subsection{Protocols under test}

We evaluate the following single-primary protocols: 
PBFT~\cite{Castro:2002:PBFT:TOCS:571637.571640}, 
SBFT~\cite{Gutea:2019:SBFT:DSN.2019.00063}, and 
MinBFT~\cite{Veronese:2013:Efficient:TC.2011.221}.
We use the variant PBFT~\cite{Castro:2002:PBFT:TOCS:571637.571640}
that uses MACs that are computationally cheaper than signatures.
SBFT~\cite{Gutea:2019:SBFT:DSN.2019.00063} provides fast-path commitment 
using $3f+c+1$ replicas out of $3f+2c+1$ replicas and linear communication. 
Chained Hotstuff~\cite{Yin:2019:HotStuff:3293611.3331591} is a rotating-primary 
protocol that changes its view for each proposal, 
and pipelines protocol messages as well as commit decisions.

Flexible MinBFT is evaluated with $f$ failures among $N=3f+1$ replicas, thus 
the normal commit quorums are $f+1$ while the view-change quorums are $2f+1$.
We evaluate both the single-chain and multi-chain variants of \thesystem, 
namely \thesystem, and \themcsystem.

We also apply \thesystem in the context of multi-primary paradigm leveraging 
the RCC~\cite{Gupta:2021:RCC:ICDE} in order to allow each replica to act as 
primary. 
With this approach, replicas can use all their resources effectively and 
provide better throughput over single-primary solutions.
We evaluate the RCC variant of \themcsystem with 
RCC-PBFT and MirBFT~\cite{Stathakopoulou:2019:MirBFT:arxiv:1906.05552}.

We implemented all the protocols within a common framework written in Go.
The framework uses gRPC~\cite{grpc} for communication and 
protobuf~\cite{protobuf} for message serialization. 
The ECDSA~\cite{Johnson:2001:ECD:2701775.2701951} algorithms in Go's 
crypto package are used for 
authenticating the messages exchanged by the clients and the replicas.
The trusted component, namely USIG~\cite{Veronese:2013:Efficient:TC.2011.221},
was implemented in C using the Intel SGX SDK~\cite{Costan:2016:SGX:TR}. 
We implemented two variants of USIG. 
For MinBFT and Flexible MinBFT, the signatures were computed using the ECDSA 
algorithm.
For \thesystem, the signatures were computed using Ed25519 signature scheme 
~\cite{Bernstein:2012:edd25519:JC:Springer} 
that supports batch verification to facilitate linear communication pattern 
(see Figure~\ref{fig:duobft:linear}).

Using our own implementation ensures a consistent evaluation of 
all protocols. 
Moreover, the source code for RCC and MinBFT were not publicly available at the 
time of evaluation. 
The publicly available Hotstuff implementation only 
proposed command hashes~\cite{Stathakopoulou:2019:MirBFT:arxiv:1906.05552}, 
whereas our implementations propose the actual payload.
The evaluation uses a key-value store benchmark because 
it serves as a good abstraction for building higher level systems
~\cite{Gutea:2019:SBFT:DSN.2019.00063}.

\subsection{Experimental setup}

We deployed the protocols using the SGX-enabled DC8v2 
virtual machines (8 vCPUs and 32GB of memory) available in the 
Microsoft Azure cloud platform~\cite{azure}. 
The virtual machines were evenly spread across ten different 
geographical regions.
The regions were East US, West US, South Central US, Canada Central, 
Canada East, UK South, North Europe, West Europe and South East Asia.
The round-trip latencies were under 30ms between regions in North America, 
under 150ms between regions in Europe and North America, and around 240ms 
between Canada and South East Asia.
We obtained multiple VMs per region and organized them into a 
Kubernetes cluster.
Each replica pod was deployed in its own VM, while multiple clients pods 
were deployed per VM.
The primary role was assigned to a replica in the East US region.

The clients are spread equally across all regions, and 
they send requests to the replicas in a closed-loop, i.e. clients wait 
for the response before sending the next request. 
We measured the throughput and latency for each of the protocols. 
The payload size is set at 512 bytes. 
Unless otherwise stated, the batch size defaults to 200 commands per batch. 
We evaluated \thesystem and \themcsystem by varying the ratio 
of Hybrid and BFT commit responses received by the client.
The suffix in the legend indicates the percentage of hybrid commit responses.

\begin{figure}
	\centering
	\begin{subfigure}[b]{0.49\columnwidth}
		\centering
		\includegraphics[width=\textwidth]{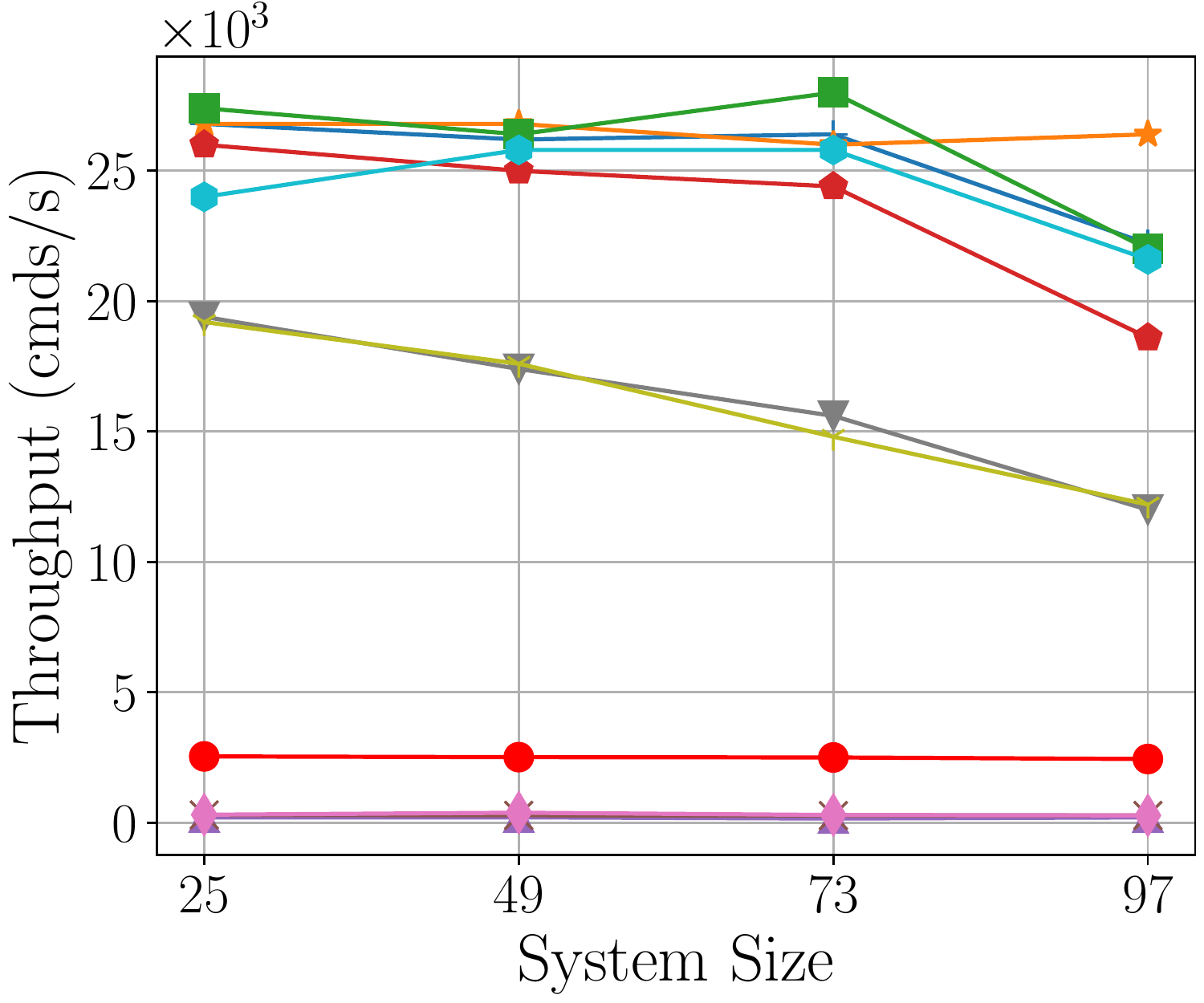}
	\end{subfigure}
	\begin{subfigure}[b]{0.49\columnwidth}
		\centering
		\includegraphics[width=\textwidth]{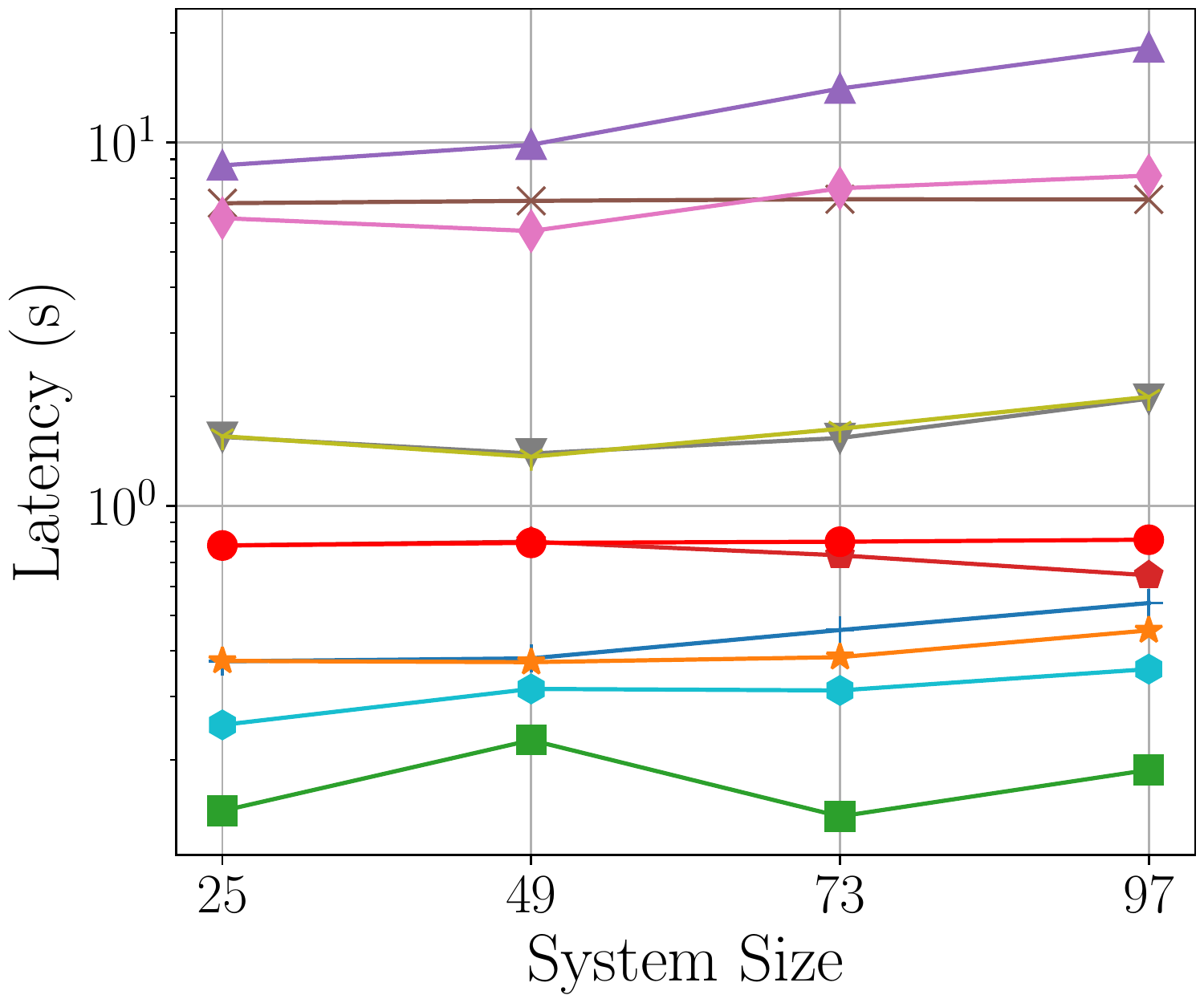}
	\end{subfigure}
	\caption{Performance under $f$ failures.}
	\label{fig:failures-scalability}
\end{figure}

\subsection{Batching Experiment}

We measured the impact of command batch size on the performance of the 
protocols.
Batching amortizes the cost of consensus by having proposing multiple commands 
together in a single block, but it also increases the network consumption on 
the primary since it should multicast the batch to all replicas.
For this experiment, we deployed 49 replicas, 
and varied the number of commands per batch between 10 and 400 and 
measure the throughput and latency for each of the protocols. 
All protocols tolerate 16 failures except MinBFT that tolerates 24 failures.
The results are in Figure~\ref{fig:batching-plot}.

Due to the lack of out-of-order processing in chain-based protocols including 
\thesystem and Chain Hotstuff, replicas cannot pipeline 
multiple proposals at the same time.
Thus, their throughput tend to be very low compared to other protocols. 
On the other hand, \themcsystem protocols leverages multiple instances to 
boost throughput and thus is able to compete better with other 
single-primary protocols. 
Since all single-primary protocols are able to process multiple  
command batches at the same time, they perform as fast as their primary 
replicas are able to disseminate batches.

Flexible MinBFT provides the lowest latency among all protocols, due to its 
two phase commit protocol with only $f+1$ quorum.
\themcsystem due to its linear communication exhibits slightly 
higher latency but within 200-milliseconds. 
\themcsystem 100\% provides 25\% and 50\% lower latency than PBFT 
and SBFT without sacrificing throughput.
Due to additional communication steps as a result of linear message patterns,
\themcsystem with 50\% and 0\% hybrid commit responses incur a latency penalty 
compared to PBFT. 
Furthermore, the overhead of collecting multiple quorums per command batch 
affects the throughput of \themcsystem 50\% and 0\%. 
Thus, these two protocols reach their peak throughput only at batch size 400, 
and incur a 10\% throughput hit compared to PBFT.

\subsection{Scalability Experiment}
\label{sec:eval:scalability}

Next, we measured the performance of the protocols as the system size is 
increased. 
In this experiment, we measure the performance of the protocols at four 
system sizes: 25, 49, 73, and 97, tolerating 8, 16, 24, and 32 failures. 
MinBFT tolerates 12, 24, 36, and 48 hybrid failures, respectively. 
The batch size is set to 200 because most of the protocols reached their peak 
throughput at this batch size in the previous experiment.
The result is in Figure~\ref{fig:scalability-plot}.

Similar to the previous experiment, 
the chain-based protocols yield low performance due to lack of out-of-order 
processing. 
\themcsystem makes up for the performance impact by 
using multiple chains, which allows it to provide similar throughput as 
other protocols. 
\themcsystem 100\% is able to provide sub-200-milliseconds latency 
even at 97 replicas. 
While \themcsystem's throughput is at least 10\% lower than other protocols 
at 25 replicas, it scales better as system size increases and performs 
at par with other protocol starting at 73 replicas.
Furthermore, its latency is at least 30\% lower than PBFT at all system sizes.
This shows that with \themcsystem, it is possible to provide 
low-latency commits under the hybrid model at scale.

On the other hand, \themcsystem incurs both throughput and latency penalty when 
providing BFT commits. The throughput overhead was around 5\% while the 
latencies were 30\% higher due to an additional communication round-trip.

\subsection{Failure Experiment}

We also measured the impact of minority $f$ failures on performance at scale. 
To do so, we repeated the previous experiment but in the presence of $f$ 
failures evenly spread across all regions.
The result is in Figure~\ref{fig:failures-scalability}.

Since all protocols must visit additional regions to collect quorums, the 
latencies of all protocols are higher. 
The throughput gap between \themcsystem 100\% and other protocols 
narrows at smaller system sizes, while the gap is larger for 50\% and 0\% cases. 
This is because for BFT commits replicas need to gather votes from all the 
regions, which slows down \themcsystem due to their additional communication 
steps. Since the hybrid commits only needs half of replicas, the throughput 
of \themcsystem 100\% is same as other protocols at all system sizes 
despite its overheads.
Furthermore, Flexible MinBFT provides the lowest latency due to its two-step 
protocol without any overheads unlike \themcsystem. 

\subsection{Multi-Primary Experiment}

\begin{figure}
	\centering
	\begin{subfigure}[b]{\columnwidth}
		\centering
		\includegraphics[width=\textwidth]{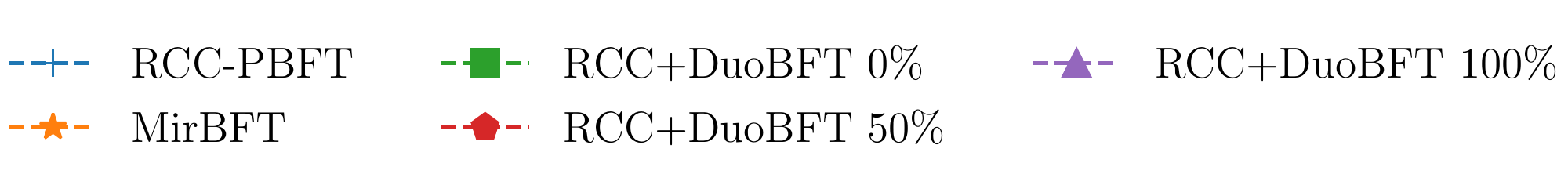}
	\end{subfigure}
	\begin{subfigure}[b]{0.49\columnwidth}
		\centering
		\includegraphics[width=\textwidth]{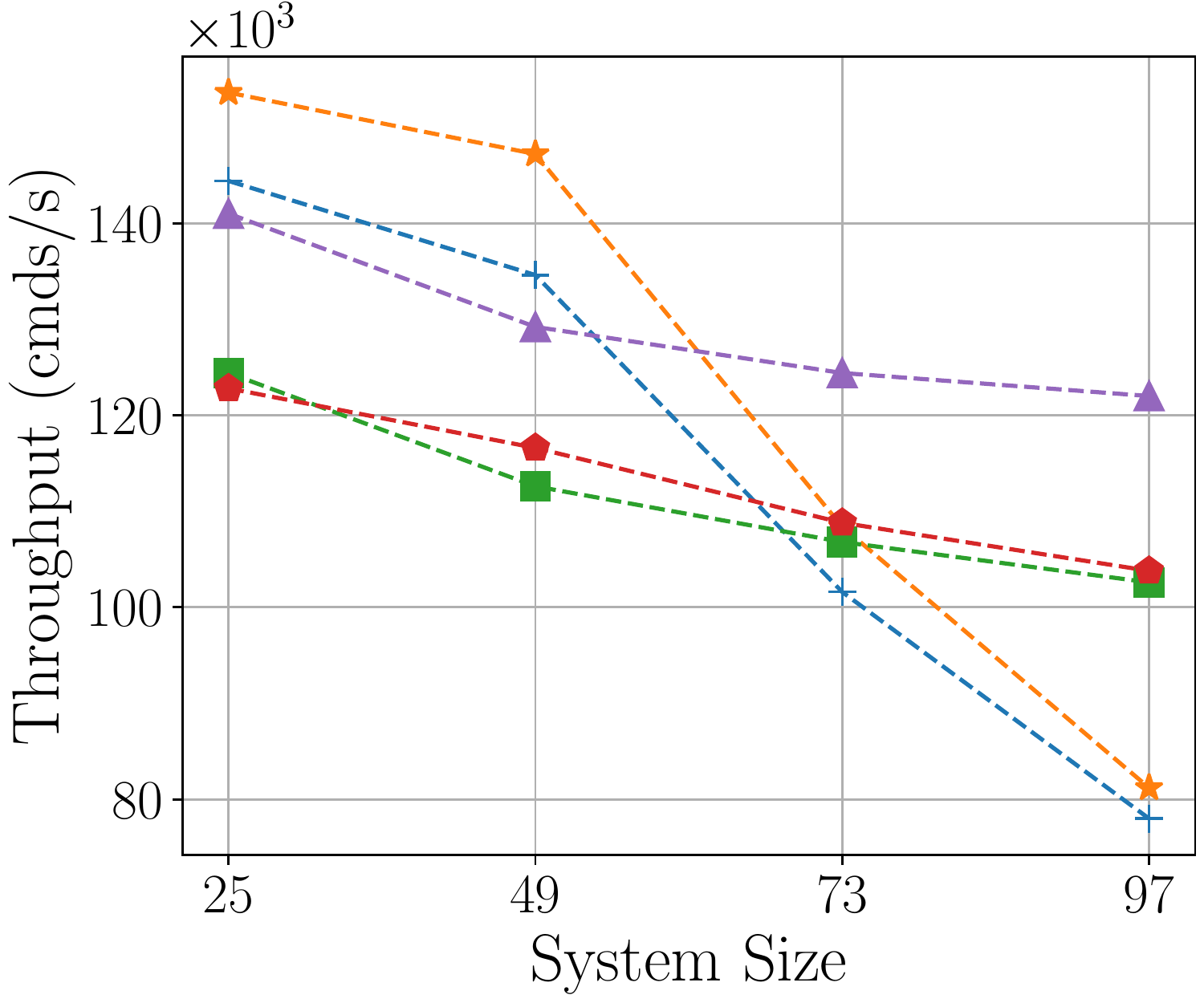}
	\end{subfigure}
	\begin{subfigure}[b]{0.49\columnwidth}
		\centering
		\includegraphics[width=\textwidth]{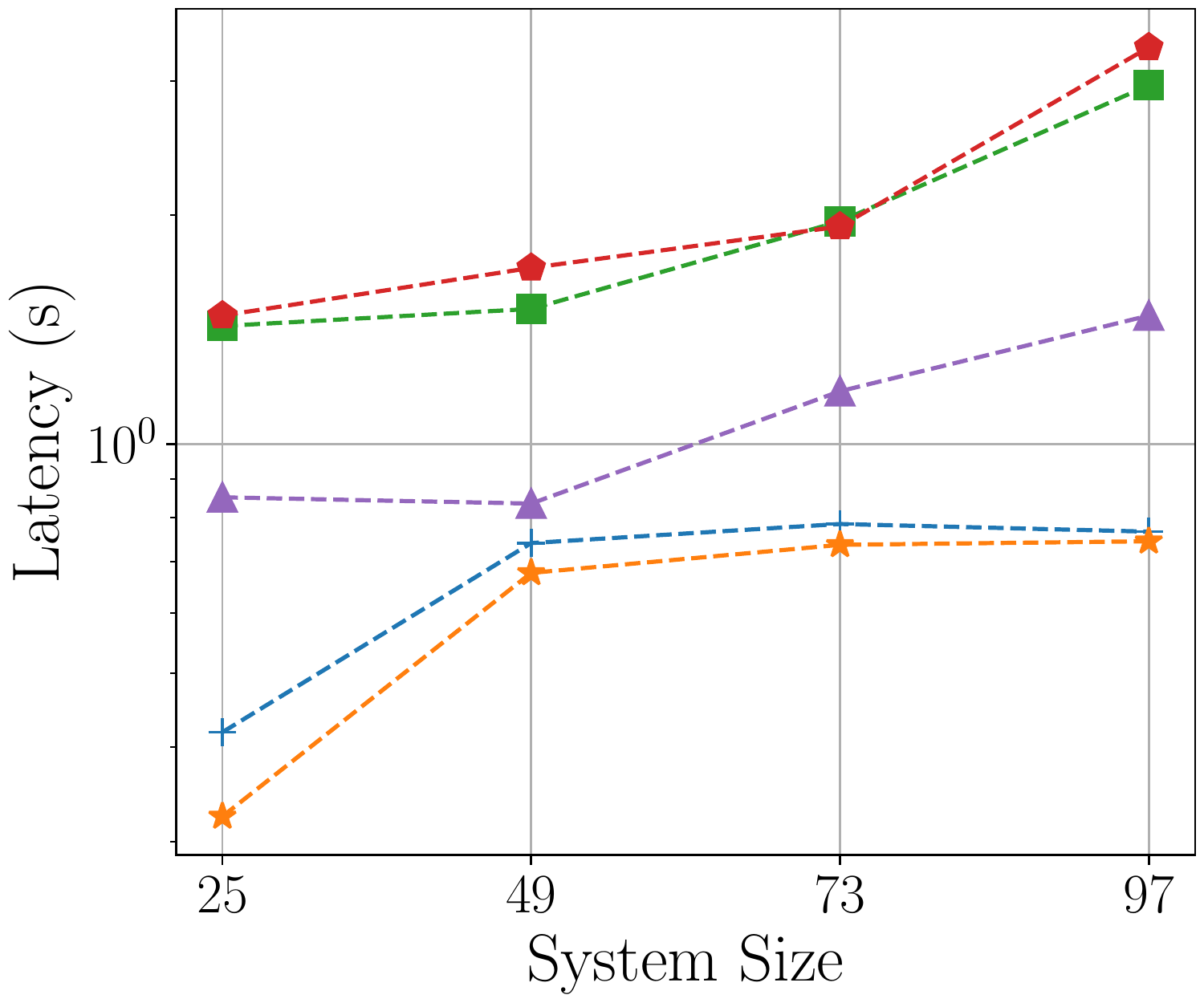}
	\end{subfigure}
	\caption{Multi-primary Performance versus System Size.}
	\label{fig:rcc-scalability}
\end{figure}

In the previous experiments, we showed that the \themcsystem protocol, at 
least with hybrid commits,
is competitive with other single-primary protocols in terms of throughput
and provide better latency than single-primary protocols.
We ran a separate experiment to observe whether \themcsystem can 
improve throughput in the context of multi-primary paradigms at scale.
The experiment parameters are same as described in 
Section~\ref{sec:eval:scalability}.
We evaluated the RCC paradigm with PBFT and \themcsystem, and MirBFT.
The results are in Figure~\ref{fig:rcc-scalability}.

We observed that \themcsystem 100\% takes advantage of the 
linear communication along with one round-trip commit and scales better 
than any other protocol. 
At 97-replicas, \themcsystem 100\% provides 1.5$\times$ more throughput than 
RCC-PBFT and MirBFT.
\themcsystem with 50\% and 0\% Hybrid commits perform better by up to 30\% 
than RCC-PBFT at higher scale by taking advantage of the linear communication. 
By distributing the responsibility of collecting signatures to different 
replicas, the RCC paradigm and the linear communication pattern 
distributes any compute and network bottlenecks. 
Thus, protocols scale better than protocols with quadratic message complexity.
We also observed higher latencies for \thesystem protocols 
since the RCC paradigm optimizes for throughput. 
Since, all instances must complete a round before the commands belonging to the 
round can be executed, the commands are usually executed at the speed in which 
the slowest/farthest replica commits its commands.
We also noticed that \themcsystem's signature computation overheads increased 
linearly with the number of replicas, and it contributed to elevated latencies 
at higher system sizes. 
We believe replacing the Ed25519 signatures with threshold signatures such as 
BLS~\cite{stathakopoulous2017threshold} signatures can help lower the latency.


\section{Related Work}
\label{sec:relwork}

Since Lamport et al. formulated the Byzantine Generals problem
\cite{Lamport:1982:ByzGenerals:TOPLS:357172.357176}, numerous solutions to 
solve the agreement problem in the face of Byzantine failures have been 
proposed. 
These solutions have varied widely in terms of their fault assumptions 
and the timing models.
A review of these solutions is beyond the scope of the paper. 
We defer the interested readers to books on distributed computing
\cite{Lynch:1996:DistAlgBook,Cachin:2011:RSDP}.

The literature is rich in protocols that adopt the hybrid fault model
\cite{Liu:2019:FastBFT:TC.2018.2860009}. 
While early protocols depended on an attested append-only log abstraction 
provided by the trusted component~\cite{Chun:2007:A2M:SOSP:1294261.1294280}, the 
counter-based abstraction~\cite{Levin:2009:TrInc:NSDI:1558977.1558978} 
became popular due to its simplicity, and have been adopted by numerous 
protocols~\cite{Liu:2019:FastBFT:TC.2018.2860009,Behl:2017:Hybster:Eurosys:3064176.3064213,Veronese:2010:EEB:1909626.1909800,Veronese:2013:Efficient:TC.2011.221,Kapitza:2012:CheapBFT:Eurosys:2168836.2168866}.
We used MinBFT to perform our analysis and construction, 
because of its presentation as the hybrid counterpart to PBFT, 
and its use of a simple counter-based trusted attestation mechanism.

\textbf{Speculation.} Some partially-synchronous
BFT protocols
\cite{Gutea:2019:SBFT:DSN.2019.00063,Kotla:2007:Zyzzyva:SOSP:1294261.1294267, Arun:2019:ezBFT:ICDCS.2019.00063,Martin:2006:FaB:TDSC.2006.35} 
use speculation to make commit decisions using fewer communication steps. 
These protocol adopt fast BFT quorums that are at least $50\%$ larger than 
normal BFT quorums and reduce communication delays under favorable 
conditions.
At times, a fast quorum may not be enough or unattainable, in which case, the 
replicas fall back to a slow protocol with normal quorums and additional 
communication steps. 
Thus, such protocol collect larger quorum than normal protocols in the best 
case, and spend more communication steps than normal protocols in the worst 
case.
\thesystem's provides a cheaper and stronger alternative to these solutions. 
The hybrid commit rule can be used to make commit decisions on a block 
by collecting votes from quorums of $1/3$ replicas. 
For the BFT commit rule, a replica only needs to collect additional 
$1/3$ votes for the block and $2/3$ votes for the next block. In total, 
a replica does not need to collect more votes or spend more communication steps 
than that in normal BFT protocols.

Similarly, Thunderella~\cite{Shi:2018:Thunderella} and 
Sync Hotstuff~\cite{Abraham:2019:SyncHotStuff:270} 
commit optimistically under the partial synchrony model using quorums of size 
$\geq 3/4$, and fallback to a synchronous slow commit rule.
\thesystem's guarantees are completely based in the partial synchrony model.

SACZyzzyva~\cite{Gunn:2019:SACZyzzyva:SRDS47363.2019.00024} uses a trusted 
counter to increase the resilience of Zyzzyva to tolerate $f$ slow nodes. 
The protocols requires only $f+1$ replicas to host the trusted 
counters in a system with $3f+1$ replicas.
While \thesystem requires all replicas to host the trusted component, 
it allows any replica to be the primary.  
Moreover, with the Flexible Hybrid Quorums approach, 
we also provide the flexibility to choose 
$f$ independently of $N$, the number of replicas.

\textbf{Tentative Execution and Optimistic Agreement.} 
Some protocols like PBFT and MinBFT use tentative execution
\cite{Castro:2001:PBFT:Thesis} to execute the 
proposed operations before the final commit step. 
This can improve the overall performance under favorable conditions.
Furthermore, the optimistic agreement~\cite{Distler:2016:ReBFT:TC.2015.2495213}
technique uses only a subset of replicas to run agreement, while the remaining 
replicas update their state passively. 
Such techniques are orthogonal and applicable to \thesystem as well.

Some benign protocols~\cite{Zhang:2018:Tapir:TOCS:3269981} 
enable clients to obtain a speculative response that sometimes deviates from 
the final response due to replica crashes or network faults.
Similarly, \thesystem's hybrid commits can deviate from BFT commits during 
trusted component compromises.

\textbf{Hierarchical Protocols.}
Steward~\cite{Amir:2010:Steward:TDSC.2008.53,Amit:2006:Steward:DSN.2006.63} 
and GeoBFT~\cite{Gupta:2020:ResilientDB:3380750.3380757} 
follow the hierarchical fault model by using a combination of crash and 
byzantine fault tolerant mechanisms.
The replicas are divided into groups. 
Replicas with a group run a BFT protocol while inter-group agreement is 
achieved using a crash-fault tolerant protocol. 
However, the protocol exposes a single combined fault model to the learners: 
the protocols can tolerate $fz$ failures in $z$ groups of which at most $f$ 
can happen in a single group.
Such techniques are aimed towards WAN deployments.

\textbf{Flexible Quorums.} 
In Flexible Paxos~\cite{Howard:FlexiblePaxos:LIPIcs:2017:7094}, Howard et al. 
introduced the notion of flexible quorums in the crash fault model.  
Malkhi et al. then developed the flexible quorums approach in the Byzantine 
fault model.
\cite{Malkhi:2019:FlexibleBFT:CCS:3319535.3354225}.
The Flexbile Hybrid Quorums presented in this paper can be seen as the 
hybrid variant of the flexible quorum technique. 
\thesystem adds supports for the BFT fault model over the hybrid fault model 
and exposes the choice to the learners.

\textbf{Adversaries.} Different kinds of adversaries have been explored in 
prior works.
Both the BAR and a-b-c fault models
\cite{Malkhi:2019:FlexibleBFT:CCS:3319535.3354225,Aiyer:2005:BAR:1095809.1095816} 
consider an adversary that does not collude. 
With the hybrid fault model, we consider an adversary that does not break the 
protections around the trusted component, but they can otherwise collude with 
other Byzantine replicas.

\textbf{Diverse learners.} Bitcoin~\cite{Nakamoto:2008:Bitcoin} uses a 
probabilistic commit rule that depends on the depth of the confirmation. 
Typically, a block depth of six implies a commit with a very high probability, 
although a block depth of one is enough to commit for some learners.
The Cross fault-tolerant (XFT)~\cite{Liu:2016:XFT:OSDI:3026877.3026915} model 
offers two kinds of learners: learners that follow the crash fault model 
under the asynchronous timing model, or learners that from the Byzantine fault
model under the synchronous timing model.

\section{Conclusion and Future Work}\label{sec:conclusion}

In this paper, we present \thesystem, a BFT protocol that provides 
commits under two fault models -- BFT and Hybrid --
under the partial synchrony timing model.
The clients can wait for responses from either or both the hybrid and BFT 
commit.
\thesystem uses the Flexible Hybrid Quorum technique to provide cheap 
Hybrid commits with $f+1$ replicas and has a unified view-change mechanism 
for both fault models.
Our experimental evaluation show that \themcsystem, the multi-chain variant of 
\thesystem, is able to provide up to 30\% lower latency than state-of-the-art 
protocols with comparable throughput.
Furthermore, \themcsystem is compatible with recent multi-primary paradigms 
and provides better scalable throughput than existing protocols.
These characteristics make \thesystem a better fit for applications that almost 
always only requires hybrid commits with a small percentage of BFT commits.

\bibliographystyle{ACM-Reference-Format}
\bibliography{bft,cft,tools,mypapers}

\appendix

\section{Flexible MinBFT Correctness}\label{sec:minbft:proof}

Note that our proof structure overlaps with MinBFT's original proof for ease 
of exposition.

\begin{lemma}\label{lm:minbft:s1}
    In a view $v$, if a correct replica executes an operation $o$ with 
    sequence number $i$, no correct replica will execute $o$ with sequence 
    number $i' \neq i$.
\end{lemma}
\begin{proof}



    If $r$ executes $o$ with sequence number $i$, then it will have $f+1$ valid 
    \texttt{Commit} messages for $\langle o, i \rangle$ from a quorum $Q_c$.

    We prove by contradiction. Suppose another correct server $r'$ execute $o$ 
    with sequence number $i' > i$. This can happen if $s'$ received $f+1$ valid 
    \texttt{Commit} messages for $\langle o, i' \rangle$ from quorum $Q_c'$. 
    Note that $\exists Q_c, Q_c' : Q_c \cap Q_c' = \emptyset$ i.e. any two 
    $\mathcal{Q}_c$ quorums need not intersect. 
    There are two cases to consider depending on the primary:
    \begin{enumerate}
        \item \textit{Primary is correct:} This is trivial since a correct 
        primary will not generate two UIs for the same operation $o$.
        \item \textit{Primary is faulty:} Lets say $r$ sends 
    $\langle \texttt{Commit}, v, r, s, UI_p, UI_r \rangle$ to $s$ for
    $\langle o, i \rangle$ and say $r'$ sends 
    $\langle \texttt{Commit}, v, r', s', UI_p, UI_r \rangle$ to $s$ for
    $\langle o, i' \rangle$.
        \begin{enumerate}
            \item $s'$ executed some operation at $i$: Primary cannot generate 
        two messages with the same sequence number. Thus, $i$ must have been 
        $o$. Since $o.seq \leq V_{req}[c]$, $o$ will not be executed again.
            \item $s'$ did not execute $i$: Replicas can only execute in 
        sequence number order. Since $i < i'$, $s'$ must wait to execute $i$. 
        Once its executes $i$, executing $i'$ fall under previous case.        
        \end{enumerate}
    \end{enumerate}
    Thus, it is not possible for any two replicas to execute the same operation 
    with different sequence number in view $v$.
\end{proof}

\begin{lemma}\label{lm:minbft:s2}
    If a correct replica executes an operation $o$ with sequence number $i$ in 
    a view $v$, no correct replica will execute $o$ with sequence number 
    $i' \neq i$ in any view $v' > v$.
\end{lemma}
\begin{proof}
    If a correct replica $s$ executes $o$ with sequence number $i$ in view $v$ 
    it must have received at least $f+1$ valid \texttt{Commit} messages for 
    $\langle o, i, v \rangle$ from quorum $Q_c$ replicas.

    Proof by contradiction: Let us suppose that another correct replica $s'$ 
    executes $o$ at sequence number $i' > i$ in $v' > v$, then $s'$ would have 
    received $f+1$ valid \texttt{Commit} messages for 
    $\langle o, i', v' \rangle$ from quorum $Q_c'$ replicas.

    Note that we have that any two commit quorums need not intersect i.e. 
    $Q_c \cap Q_c' = \emptyset$.

    We first deal with the case where $v' = v+1$ and generalize later for 
    arbitrary values of $v' > v$.

    Let $p$ be the primary of new view $v'$. First, we show that the primary 
    $p$ in view $v'$ cannot deny the fact that $o$ was accepted/executed 
    before $v'$. Then, we show that no correct replica will execute $o$ with 
    $i' \neq i$ in $v'$ and prove the contradiction.

    The \texttt{NewView} message sent by the new primary includes the new view 
    certificate $V_{vc}$ that contains $N - f$ \texttt{ViewChange} 
    messages from a quorum $Q_{vc}$ that contains of at least 
    one correct replica, say $r \in Q_{vc}$ that will send the correct 
    \texttt{ViewChange} message to the new primary. 
    Given this, we consider the following cases:

    \begin{enumerate}
        \item\label{proof:lm2:case1} \textit{Primary is correct and replica 
$r \in Q_{vc}$ is correct}: If primary $p$ is correct, it inserts $N-f$
\texttt{ViewChange} messages into $V_{vc}$ including the one from $r$. There are
two possibilities.
        \begin{enumerate}
            \item $o$ was executed after checkpoint: $r$ is correct, so $O$ 
        contains \texttt{Commit} that $r$ sent for $o$, therefore $V_{vc}$ and 
        $S$ in the \texttt{NewView} message assert $O$ was executed explicitly.
            \item $o$ was executed before the checkpoint: The latest checkpoint 
        $C_l$ shared by $r$ implies that $o$ was executed implicitly. Since 
        the $V_{vc}$ sent in the \texttt{NewView} contains the $C_l$, it 
        implies that $o$ was executed.
        \end{enumerate}

        \item\label{proof:lm2:case2} the primary $p$ is correct, but there is 
        no correct replica in 
        $Q_{vc}$ that executed $o$: There should exist at least a faulty replica 
        $r \in Q_c'$ that accepted $o$ because $Q_c' \cap Q_{vc}' \neq \emptyset$ 
        ($q_{cmt} + q_{vc} > 1$).
        \begin{enumerate}
            \item $o$ was executed after checkpoint: $r$ might be tempted to not
        include $o$'s messages in $O$, but if it did, $p$ being correct would 
        not put $r$'s \texttt{ViewChange} message in $V_{vc}$ as $p$ can 
        detect its invalid. There are two possible ways a detection will 
        happen: (i) if $r$ executed a request 
        $o'$ after $o$, $r$ might put the \texttt{Commit} message for $o'$ but 
        not $o$ leaving a hole in the log that $p$ will detect. (ii) if $r$ sent 
        \texttt{Commit} for $o$ with a USIG value $cv$, it might leave out all 
        commit after $o$ with $cv' > cv$ from the $O$ log. But, this log will 
        also be considered 
        invalid by $p$ since $r$ must sign the \texttt{ViewChange} message 
        containing $O$ before sending it. The USIG will sign with a 
        $cv'' > cv+1$ that will
        allow a correct $p$ to detect an incomplete $O$. Thus, $r$ must include 
        all commits and Case~\ref{proof:lm2:case1} above will apply.
            \item $o$ was executed before the stable checkpoint: One way this 
        can happen is when $r$ includes an older checkpoint message but $p$ 
        will detect the invalid \texttt{ViewChange} message because the USIG 
        value of the message will disclose that there are messages since the 
        checkpoint message that $r$ failed to disclose.
        \end{enumerate}
    
        \item\label{proof:lm2:case3} Primary is faulty but $r \in Q_{vc}$ is 
        correct and executed $o$. In 
        this case, the faulty primary $p$ may attempt to modify the contents of 
        $O$ that it receives from $r$ before inserting into $V_{vc}$. 
        However, this will leave a hole and other correct replicas will detect 
        this misbehavior since they run the same procedure the primary runs
        for computing the \texttt{NewView} message. If $p$ removes $o$ and 
        all further operations after $o$, 
        correct replicas can also detect it because the USIG value of the 
        $r$'s \texttt{ViewChange} message inside $V_{vc}$ will indicate 
        the missing messages (as in Case~\ref{proof:lm2:case2}). Similarly,
        if the primary tries to add an older checkpoint certificate, correct 
        replicas will detect it from the holes in the USIG values. Therefore, 
        a fault primary cannot tamper with a \texttt{ViewChange} message without 
        detection. Thus, Case~\ref{proof:lm2:case1} will happen.

        \item Primary is faulty and no correct $r \in Q_{vc}$ has executed $o$. 
        A faulty $r \in Q_{vc}$ may exist. Given $|Q_{c}| + |Q_{vc}| > N$, 
        $r$ cannot successfully convince the primary to behave as if it did not 
        execute $o$. Even if the primary being faulty uses $r$'s 
        \texttt{ViewChange} message in the $V_{vc}$, other replicas will detect
        the missing sequence number and the corresponding commit message for 
        $o$. Thus, we will fall back to 
        Cases~\ref{proof:lm2:case2} and \ref{proof:lm2:case3} above.

    \end{enumerate}

    The above four cases show that $p$ in view $v'$ must assert $O$ was 
    executed before $v'$ in the certificate $V_{vc}$. Now, we show no correct 
    replicas will execute $o$ with $i' \neq i$ in $v'$. There are two cases:

    \begin{enumerate}
        \item \textit{Primary is correct:} A correct primary $p$ will never 
        generate a second USIG certificate for the same operation and correct 
        replicas will not send a commit message $o$ with $i'$ in view $v'$.
        \item \textit{Primary is faulty:} It is possible for a faulty primary  
        to create a new \texttt{Prepare} message for $o$ and successfully 
        create a new USIG $UI_p' = \langle i, H(o) \rangle_p$ and send it to 
        a replica $r$. However, 
        every replica maintains the $V_{req}$ that holds the last executed 
        operation identifier $seq$ for each client. Thus, $r$ will discover 
        that $o$ was already executed since $o.seq \leq V_{req}[c]$. Thus,
        $o$ will not be executed again.
    \end{enumerate}

    This proves that if a correct replica executed $o$ at sequence number 
    $i$ in view $v$, then no correct replica will execute $o$ at sequence 
    number $i' \neq i$ in view $v' = v+1$.

    We now generalize for arbitrary values of $v' > v$. There are two cases:
    \begin{enumerate}
        \item \textit{$v'=v+k$ but no request was accepted in view $v''$ such 
that $v' < v'' < v+k$}: This case is trivial and falls under the case of 
$v' = v+1$ since only view change related messages are sent in $v''$ which 
mirrors the $v$ to $v'$ transition.
        \item \textit{$v'=v+k$ but requests were prepared/accepted in view 
$v''$ such that $v' < v'' < v+k$}: At each view change, replicas must propagate 
information about operations from one view to its consecutive view (e.g. $v$ 
to $v+1$, and so on). This is done either via the checkpoint certificate or the 
via the $O$ log set. Thus, each transition becomes the case of $v' = v+1$ above.
    \end{enumerate}

\end{proof}

\begin{theorem}\label{th:minbft:s1}
    Let $s$ be a correct replica that executed more operations of all correct 
    replicas up to a certain instant. If $s$ executed the sequence of 
    operations $S = \langle o_1, ... o_i \rangle$, then all other correct 
    replicas executed this same sequence of operations or a prefix of it.
\end{theorem}
\begin{proof}
    Let $prefix(S, k)$ be a function that gets the prefix of sequence $S$ 
    containing the first $k$ operations, with $prefix(S, 0)$ being the empty 
    sequence. Let $\bullet$ be an operation that concatenates sequences.

    We prove by contradiction.
    Assume the theorem is false, i.e there exists a correct replica $r'$ that 
    executed some sequence of operations $S'$ that is not a prefix of $S$.
    Let $prefix(S, i) = prefix(S', i-1) \bullet \langle o_i \rangle$ and 
    $prefix(S', i) = prefix(S, i-1) \bullet \langle o_i' \rangle$ such that 
    $o_i \neq o_i'$
    In this case, $o_i$ was executed as the $i$th operation by replica $r$ and 
    $o_i'$ was executed as the $i$th operation by replica $r'$. 
    Assume $o_i$ was executed in view $v$ and $o_i'$ was executed in view $v'$.
    Setting $v = v'$ will contradict Lemma~\ref{lm:minbft:s1} and setting
    $v' \neq v$ will contradict Lemma~\ref{lm:minbft:s2}.

\end{proof}

\begin{lemma}\label{lm:minbft:l1}
    During a stable view, an operation requested by a correct client completes.
\end{lemma}
\begin{proof}
    A correct client $c$ will send an operation $o$ with an identifier larger 
    than any previous identifiers to the replicas. The primary $p$ being 
    correct, will construct a valid \texttt{Prepare} message with a valid 
    USIG certificate $UI_p = \langle i, H(o) \rangle_p$ and send it to all 
    replicas. At least $f+1$ correct replicas will validate the \texttt{Prepare}
    message, verify the $UI$, and send a corresponding \texttt{Commit} message.
    Since there can be only $f$ faults, there should exist at least $N - f$ 
    correct replicas, out of which $f+1$ ($q_{cmt}$), should successfully 
    produce these \texttt{Commit} messages. When a correct replica receives 
    $q_{cmt}$ valid \texttt{Commit} messages, $o$ will be executed and replied
    to the client $c$. Since $q_{cmt}$ correct replicas exist, a correct client 
    will receive $f+1$ same replies indicating that operation $o$ was properly 
    executed at sequence number $i$.
\end{proof}

\begin{lemma}\label{lm:minbft:l2}
    A view $v$ eventually will be changed to a new view $v' > v$ if at least 
    $N-f$ correct replicas request its change.
\end{lemma}
\begin{proof}
    A correct replica $r$ sends a $\langle ReqViewChange, r, v, v' \rangle$ 
    message requesting a view change to all replicas. However, at least $N-f$
    correct replicas must send such a message to actually trigger a view change.
    Say a set of $N-f$ correct replicas request a view change from $v$ to 
    $v+1$ by sending the \texttt{ReqViewChange} message. The primary for the 
    new view is $p = (v+1) \mod N$. Consider the two cases:
    \begin{enumerate}
        \item \textit{the new view is stable}: correct replicas will receive 
the \texttt{ReqViewChange} messages. Consequently, correct replicas that 
receive at least $N-f$ \texttt{ReqViewChange} messages will enter the new 
view $v'$ and send a \texttt{ViewChange} message to all replicas. The primary 
$p$, being stable, for view $v+1$ will send a valid \texttt{NewView} message 
in time. Thus, correct replica that receive the message will transition to 
new view $v' = v+1$.
        \item \textit{the new view is not stable}: We consider two cases:
        \begin{enumerate}
            \item \textit{the primary $p$ is faulty and does not send the 
\texttt{NewView} message in time, or $p$ is faulty and sends an invalid 
\texttt{NewView} message}, or $p$ is not faulty but the network delays $p$'s 
message indefinitely. In all these cases, the timer on other correct replicas 
that sent the \texttt{ViewChange} message will expire waiting for the new view 
message. These replicas will trigger another view change to view $v+2$.
            \item \textit{the primary $p$ is faulty and sends the NewView 
message to only a quorum $Q_{vc}$ of $N-f$ replicas but less than $N-f$ 
replicas are correct, or $p$ is correct but there are communication delays.} 
The replicas in quorum $Q_{vc}$ may enter the new view and process requests in 
time. However, the correct replicas that does not receive the NewView message 
will timeout and request change to view $v+2$. However, there will be less than
$N-f$ replicas, so a successful view change trigger will not happen. If the 
faulty replicas deviate from the algorithm, other correct replicas will join 
to change the view. 
        \end{enumerate}
    \end{enumerate}
\end{proof}

\begin{theorem}\label{th:minbft:l1}
    An operation requested by a correct client eventually completes.
\end{theorem}
\begin{proof}
    The proof follows from the Lemmas~\ref{lm:minbft:l1} and~\ref{lm:minbft:l2}.

    When the view is stable, Lemma~\ref{lm:minbft:l1} shows that the client 
    operations are properly committed. However, when the view $v$ is not stable, 
    there are two possibilities:
    \begin{enumerate}
        \item \textit{at least $f+1$ replicas timeout waiting for messages and 
        request a view change:} Lemma~\ref{lm:minbft:l2} handles this case and 
        ensures that a stable view $v' > v$ is established.
        \item \textit{less than $f+1$ replicas request a view change:} 
        There should exist at least a quorum $Q$ of $f+1$ replicas that are 
        in the current view $v$. 
        As long as these replicas continue to follow the algorithm, 
        they will continue to stay in view $v$ and client requests will be 
        committed in time. However, if the replicas are not timely, then the 
        correct replica from $Q$ in view $v$ will send the 
        \texttt{ReqViewChange} message. With this message, a successful view 
        change is triggered and the previous case takes happens.
    \end{enumerate}

    If the new view $v'$ is not stable, another view change will be triggered 
    depending on whether Cases 1 or 2 above holds.
    However, this process will not continue forever.
    Since there are only $f$ byzantine replicas and due to the 
    assumption that the network delays do not grow indefinitely, 
    eventually there should exist a view $v''$ that is stable such that 
    the primary responds in a timely manner and follows the algorithm. 

\end{proof}

\end{document}